\documentclass[11pt,onecolumn]{IEEEtran}
\newif\ifarXiv
\arXivtrue
\usepackage{authblk}

\usepackage[utf8]{inputenc} 
\usepackage[T1]{fontenc}
\usepackage{url}
\usepackage[cmex10]{amsmath} % Use the [cmex10] option to ensure complicance
\usepackage{mleftright}       % fix to wrong spacing of \left-,
\mleftright                   % \middle- \right-commands 

\usepackage{graphicx}         % provides \includegraphics{\ldots} to
\usepackage{booktabs}         % fixes poor spacing in tables and
\usepackage{amsmath,amssymb,amsfonts,amsthm}
\usepackage{enumitem}
\usepackage{mathtools}
\usepackage{algorithmic}
\usepackage[ruled,vlined,linesnumbered]{algorithm2e}
\SetKwInput{Input}{Input}
\SetKwInput{Output}{Output}
\SetKw{To}{ to }
\SetKw{Print}{Print }
\usepackage{textcomp}
\usepackage[table,xcdraw]{xcolor}
\usepackage{tikz}

\usepackage{subcaption}
\usepackage{multicol,multirow}
\usepackage{comment}

\usepackage[style=ieee,citestyle=numeric-comp,sorting=nty]{biblatex}
\newtheorem{theorem}{Theorem}
\newtheorem{remark}{Remark}

\newtheorem{lemma}{Lemma}
\newtheorem{definition}{Definition}

\newtheorem{corollary}{Corollary}

\newcommand{\bF}{\mathbb{F}}

\newcommand{\cA}{\mathcal{A}}

\newcommand{\cC}{\mathcal{C}}

\newcommand{\cI}{\mathcal{I}}

\newcommand{\cM}{\mathcal{M}}

\newcommand{\cQ}{\mathcal{Q}}
\newcommand{\cR}{\mathcal{R}}
\newcommand{\cS}{\mathcal{S}}
\newcommand{\cT}{\mathcal{T}}

\newcommand{\cV}{\mathcal{V}}
\newcommand{\cW}{\mathcal{W}}

\newcommand{\boldk}{\mathbf{k}}

\newcommand{\bolds}{\mathbf{s}}
\newcommand{\boldt}{\mathbf{t}}

\newcommand{\boldx}{\mathbf{x}}
\newcommand{\boldy}{\mathbf{y}}
\newcommand{\boldz}{\mathbf{z}}

\DeclarePairedDelimiter{\floor}{\lfloor}{\rfloor} % Use \floor* to scale.
\DeclareSymbolFont{bbold}{U}{bbold}{m}{n}
\DeclareSymbolFontAlphabet{\mathbbold}{bbold}
\newcommand{\1}{\mathbbold{1}}

\excludecomment{appendix}
\begin{document}

\title{Perfect Subset Privacy in Polynomial Computation via Reed-Muller Information Super-sets} 

% %%% Single author, or several authors with same affiliation:
% \author{%
%  \IEEEauthorblockN{Andrew R.~Barron}
%  \IEEEauthorblockA{Department of Statistics and Data Science\\
%                    Yale University\\
%                    New Haven, CT, USA\\
%                    Email: andrew.barron@yale.edu}
% }

	\author{\textbf{Zirui (Ken) Deng}$^\star$, \textbf{Vinayak Ramkumar}$^{\dagger}$, and \textbf{Netanel Raviv}$^\star$\\
$^\star$Department of Computer Science and Engineering, Washington University in St. Louis, St. Louis, USA\\
		$^\dagger$Department of Electrical Engineering--Systems, Tel Aviv University, Tel Aviv, Israel\\
   \texttt{d.ken@wustl.edu}, \texttt{vinram93@gmail.com}, \texttt{netanel.raviv@wustl.edu}}

\maketitle

%%%%%%
%% Abstract: 
%% If your paper is eligible for the student paper award, please add
%% the comment "THIS PAPER IS ELIGIBLE FOR THE STUDENT PAPER
%% AWARD." as a first line in the abstract. 
%% For the final version of the accepted paper, please do not forget
%% to remove this comment!
%%

\begin{abstract}
% \blue{VR: Remove Information super-sets from title as it is not a commonly used term???}
% \red{NR: Not commonly used indeed, but I believe we can be bold and try coining a new term.}
%\blue{VR: How about "Perfect Subset Privacy in Polynomial Computation via Reed-Muller Codes"? I prefer to avoid information supersets in the title as it is not a known terminology.}
Delegating large-scale computations to service providers is a common practice which raises privacy concerns.
%that are often addressed using information-theoretic tools. 
This paper studies information-theoretic privacy-preserving delegation of data to a service provider, who may further delegate the computation to auxiliary worker nodes, in order to compute a polynomial over that data at a later point in time. We study 
%aim at devising privacy-preserving 
techniques which are compatible with robust management of distributed computation systems, an area known as \textit{coded computing}. Privacy in coded computing, however, has traditionally addressed the problem of colluding workers, and assumed that the server that administrates the computation is trusted. This viewpoint of privacy does not accurately reflect real-world privacy concerns, since normally, 
%reflect how data delegation tasks normally occur, where 
the service provider as a whole (i.e., the administrator and the worker nodes) form one cohesive entity which itself poses a privacy risk. This paper aims to shift the focus of privacy in coded computing to safeguarding the privacy of the user against the service provider as a whole, instead of merely against colluding workers inside the service provider. To this end, we leverage the recently defined notion of \textit{perfect subset privacy}, which guarantees zero information leakage from all subsets of the data up to a certain size. 
%\red{The size of these subsets serves as a tunable privacy parameter that is in tradeoff with other system parameters. Perfect subset privacy is an attractive notion, since it provides a uniform guarantee to all parts of the data, that is not obtained by, say, bounding information leakage.} 
Using known techniques from Reed-Muller decoding, we provide a scheme which enables polynomial computation with perfect subset privacy in straggler-free systems. Furthermore, by studying \textit{information super-sets} in Reed-Muller codes, which may be of independent interest, we extend the previous scheme to tolerate straggling worker nodes inside the service provider. 
%In distributed computation systems, data is distributed among multiple workers by an admin, each worker performs computation on their portion of the data, and the admin then aggregates these computations and sends an answer to the user. Typically in such settings, privacy guarantees are provided against colluding workers, and the admin is assumed to be a trusted entity. In this work, we consider the problem of ensuring some level of data privacy against the admin. We propose a distributed system, based on Reed-Muller codes, that is capable of computing multivariate polynomials over privatized data in the presence of straggling workers.  
\end{abstract}
\begin{IEEEkeywords}
Perfect subset privacy, distributed systems, coded computing, Reed-Muller codes.
\end{IEEEkeywords}
\thispagestyle{empty} % To remove before submission to get rid of page numbers.
\section{Introduction}

% Papers in the form of a PDF file, formatted as described below, may be
% submitted online at
% \begin{center}
%   \url{https://cmsworkshops.com/ISIT2024/papers.php}
% \end{center}
% The deadline for registering the paper and uploading the manuscript is \textbf{January 29, 2024}. \textbf{No extensions will be given.}

% A paper's primary content is restricted in length to \textbf{5 pages}, but
% authors are allowed an optional 6th page only containing references. Authors may further submit up to 5 pages of supplementary material that will not be published but that the reviewers will have the option to read and review if they deem it necessary. The submission may contain a link to a longer, online version of the submitted manuscript, if the authors wish to include one.
 
% The IEEEtran-conference style should be used as presented here. Submissions should use a font size no smaller than 10 points and have reasonable margins on all the 4 sides of the text.

% Each paper must be classified as ``eligible for student paper award''
% or ``not eligible for student paper award''. Papers that are eligible for the student paper award must also include the comment ``THIS PAPER IS ELIGIBLE FOR THE STUDENT PAPER AWARD." as a first line in the abstract.

{\let\thefootnote\relax\footnote{ This work was presented in part at the 2024 IEEE International Symposium on Information Theory (ISIT). Parts of this work were done when V.~Ramkumar was a visiting researcher in the Department of Computer Science and Engineering in Washington University in St. Louis.}} 
%\red{[Boilerplate text about privacy.]}

% \red{[NOTE: Go over introduction and dial down the message that information super-sets is a new notion. Emphasize that similar problems have been studied in the context of punctured RM codes.]}

% \red{NOTE: A paragraph which summarizes the construction of binary information supersets (given in the appendix) should go somewhere in the paper.} \blue{Done.}

As significant computational power is no longer within the grasp of individual users, relying on third-party services for data storage and computation has become a necessity. While this paradigm allows smaller, non-technological entities to harness capabilities beyond their means, it inevitably poses risks related to the privacy of these entities.
% linking to various concerns spanning from mental health to the future of democracy. 
%Coding theory, established since the 1950s as a mathematical framework for communication in noisy channels, has seen diverse applications across various technologies in storage systems, quantum computers, blockchains, etc. 
%\red{[Introduce coded computing, and say what's the common privacy assumption.]} 

In recent years, there has been considerable interest in 
%another domain of 
applications of error-correcting codes in distributed computations, commonly known as \textit{coded computing}, in order to tackle straggling (i.e., slow) workers, adversarial workers, and privacy issues in such systems~\cite{lee2017speeding,yu2019lagrange,wang2022breaking}.
Privacy guarantees, however, have typically been given in terms of some level of restricted collusion among workers, while the system administrator is assumed to be trusted. 

%\red{[Elucidate what's the problem with the common privacy assumption in coded computing.]} 
This approach overlooks the common setting where a data owner relies on a single \textit{untrusted} service provider in order to store and compute over their data, and the service provider itself inherently introduces a privacy risk. Whether the storage or computation is distributed internally or not, the service provider has access to the entirety of the data, rendering the assumption of limited collusion among workers largely irrelevant. This scenario is widely prevalent, as many institutions delegate their data in its entirety to service providers such as AWS (Amazon Web Services) or Google Cloud. 
%, and especially given that many individuals, companies, and institutions opt not to maintain their own computational infrastructure and instead often choose to rent services from providers such as AWS (Amazon Web Services) or Google Cloud. 
In this work, we seek to extend the scope of coded computing
% to this~\textit{single-provider setup}
and develop new techniques to protect privacy against the service provider as a whole, while enabling the service provider's internal system administrator (abbrv. admin) to distribute the computation internally. %\red{[Consider adding: The user does not communicate with the workers directly, but rather through the system admin. Moreover, being part of the service provider as a whole, all workers are  completely free to collude, and assuming restricted collusion is largely irrelevant.]}

%\red{[Problem setup and goal.]} 
In our setup, we assume the user has some data~$\boldx$, which is internally partitioned to~$n$ parts~$x_1,\ldots,x_n$ of some contextual meaning (e.g., features). The service provider consists of a system admin and workers, among which collusion is not restricted. The user cannot communicate with the workers directly, but rather through the system admin. The user encodes~$\boldx$ using some random key vector~$\boldk$ and then sends the encoded (privatized) version of the data~$\Tilde{\boldx}$ to the admin, who in turn distributes~$\Tilde{\boldx}$ to multiple workers. Due to storage limitations, the user wishes to discard $\boldx$ and instead only keep the key $\boldk$ as side information. At a later point in time, the user is interested in computing some polynomial~$f$ over~$\boldx$, and sends~$f$ to the admin. 
%The goal of the service provider as a whole is to perform computations over the privatized data $\Tilde{X}$ and return the result to the user so that the user is able to utilize this outcome, together with the key $K$, to extract the desired computation. 
The goal of the admin is to reply to the user with some information~$\cA$, computed with the aid of the workers, so that the user can extract the desired computation results, potentially using~$\boldk$.
%of the service provider as a whole is to perform computations over the privatized data $\Tilde{X}$ and return the result to the user so that the user is able to utilize this outcome, together with the key $K$, to extract the desired computation. 
The purpose of this work is to shift the ``privacy barrier'' in coded computing: instead of protecting the privacy of the admin against colluding workers, we wish to protect the privacy of the user against the service provider as a whole. We call the above the \textit{single-provider setup}, illustrated in Figure~\ref{fig:barrier}.

\begin{figure}
\centering

\tikzset{every picture/.style={line width=0.75pt}} %set default line width to 0.75pt        

\begin{tikzpicture}[x=0.75pt,y=0.75pt,yscale=-1,xscale=1]
%uncomment if require: \path (0,1323); %set diagram left start at 0, and has height of 1323

%Shape: Rectangle [id:dp03401888198928793] 
\draw  (252.77,977.22) -- (318.26,977.22) -- (318.26,1053.05) -- (252.77,1053.05) -- cycle ;
%Flowchart: Process [id:dp26468846257644163] 
\draw  (410.33,975.88) -- (464.07,975.88) -- (464.07,1050.37) -- (410.33,1050.37) -- cycle ;
%Straight Lines [id:da6538804502309443] 
\draw    (327.12,1016.8) -- (392.92,1017.23) ;
\draw [shift={(394.92,1017.24)}, rotate = 180.38] [color={rgb, 255:red, 0; green, 0; blue, 0 }  ][line width=0.75]    (10.93,-3.29) .. controls (6.95,-1.4) and (3.31,-0.3) .. (0,0) .. controls (3.31,0.3) and (6.95,1.4) .. (10.93,3.29)   ;
%Straight Lines [id:da29516897294699795] 
\draw    (394.15,1046.54) -- (377.01,1046.43) -- (327.96,1046.11) ;
\draw [shift={(325.96,1046.1)}, rotate = 0.37] [color={rgb, 255:red, 0; green, 0; blue, 0 }  ][line width=0.75]    (10.93,-3.29) .. controls (6.95,-1.4) and (3.31,-0.3) .. (0,0) .. controls (3.31,0.3) and (6.95,1.4) .. (10.93,3.29)   ;
%Straight Lines [id:da5198760646887166] 
\draw    (327.12,986.83) -- (392.92,987.27) ;
\draw [shift={(394.92,987.28)}, rotate = 180.38] [color={rgb, 255:red, 0; green, 0; blue, 0 }  ][line width=0.75]    (10.93,-3.29) .. controls (6.95,-1.4) and (3.31,-0.3) .. (0,0) .. controls (3.31,0.3) and (6.95,1.4) .. (10.93,3.29)   ;
%Straight Lines [id:da28169220594543054] 
\draw    (464.65,1003.09) -- (497.36,980.34) ;
\draw [shift={(499,979.2)}, rotate = 145.18] [color={rgb, 255:red, 0; green, 0; blue, 0 }  ][line width=0.75]    (10.93,-3.29) .. controls (6.95,-1.4) and (3.31,-0.3) .. (0,0) .. controls (3.31,0.3) and (6.95,1.4) .. (10.93,3.29)   ;
%Straight Lines [id:da4769903129785589] 
\draw    (464.26,1022.98) -- (497.51,1052.86) ;
\draw [shift={(499,1054.2)}, rotate = 221.94] [color={rgb, 255:red, 0; green, 0; blue, 0 }  ][line width=0.75]    (10.93,-3.29) .. controls (6.95,-1.4) and (3.31,-0.3) .. (0,0) .. controls (3.31,0.3) and (6.95,1.4) .. (10.93,3.29)   ;
%Straight Lines [id:da7807965120470515] 
\draw    (464.65,1012.01) -- (497,1012.19) ;
\draw [shift={(499,1012.2)}, rotate = 180.31] [color={rgb, 255:red, 0; green, 0; blue, 0 }  ][line width=0.75]    (10.93,-3.29) .. controls (6.95,-1.4) and (3.31,-0.3) .. (0,0) .. controls (3.31,0.3) and (6.95,1.4) .. (10.93,3.29)   ;
%Shape: Rectangle [id:dp7652706646077754] 
\draw   (501,954.03) -- (560.95,954.03) -- (560.95,991.2) -- (501,991.2) -- cycle ;
%Shape: Rectangle [id:dp8682907327897391] 
\draw   (501,998.29) -- (561,998.29) -- (561,1033.2) -- (501,1033.2) -- cycle ;
%Shape: Rectangle [id:dp481767122127716] 
\draw   (501,1040.9) -- (560,1040.9) -- (560,1072.56) -- (501,1072.56) -- cycle ;
%Straight Lines [id:da6345164193294508] 
\draw [dashed, color={rgb, 255:red, 126; green, 211; blue, 33 }  ,draw opacity=1 ]   (400.31,949.39) -- (400.3,966.34) -- (400.24,979.27) -- (399.54,1077.4) ;
%Straight Lines [id:da6602878135559791] 
\draw [dashed, color={rgb, 255:red, 208; green, 2; blue, 27 }  ,draw opacity=1 ]   (480.15,951.17) -- (480.14,970.8) -- (480.08,983.73) -- (479.38,1075.17) ;
%Shape: Rectangle [id:dp45030648293374864] 
\draw [dotted]  (407.63,949.91) -- (567.12,949.91) -- (567.12,1074.62) -- (407.63,1074.62) -- cycle ;
%Straight Lines [id:da865378379695853] 
\draw    (275.11,1053.94) -- (275.11,1097.91) ;
\draw [shift={(275.11,1097)}, rotate = 270] [color={rgb, 255:red, 0; green, 0; blue, 0 }  ][line width=0.75]    (10.93,-3.29) .. controls (6.95,-1.4) and (3.31,-0.3) .. (0,0) .. controls (3.31,0.3) and (6.95,1.4) .. (10.93,3.29)   ;
%Straight Lines [id:da35603414444569803] 
\draw    (295.91,1053.94) -- (295.91,1097.91) ;
%\draw [shift={(295.91,1126.91)}, rotate = 270] [color={rgb, 255:red, 0; green, 0; blue, 0 }  ][line width=0.75]    (10.93,-3.29) .. controls (6.95,-1.4) and (3.31,-0.3) .. (0,0) .. controls (3.31,0.3) and (6.95,1.4) .. (10.93,3.29)   ;
\draw [shift={(295.91,1097)}, rotate = 270] [color={rgb, 255:red, 0; green, 0; blue, 0 }  ][line width=0.75]    (10.93,-3.29) .. controls (6.95,-1.4) and (3.31,-0.3) .. (0,0) .. controls (3.31,0.3) and (6.95,1.4) .. (10.93,3.29)   ;

%%%
%Shape: Rectangle [id:dp4080713775327207] 
%\draw   (261.24,1127.53) -- (309.01,1127.53) -- (309.01,1159.2) -- (261.24,1159.2) -- cycle ;
\draw   (261.24,1097.53) -- (309.01,1097.53) -- (309.01,1126.2) -- (261.24,1126.2) -- cycle ;

% Text Node
\draw (355.64,965.79) node [anchor=north west][inner sep=0.75pt]  [font=\normalsize] [align=left] {$\displaystyle \Tilde{\boldx}$};
% Text Node
\draw (412.91,1005.23) node [anchor=north west][inner sep=0.75pt]   [align=left] {{\fontfamily{ptm}\selectfont {\large Admin}}};
% Text Node
\draw (265.51,1006.46) node [anchor=north west][inner sep=0.75pt]   [align=left] {{\fontfamily{ptm}\selectfont {\Large User}}};
% Text Node
\draw (353.1,1028.73) node [anchor=north west][inner sep=0.75pt]  [font=\large] [align=left] {$\displaystyle \mathcal{A}$};
% Text Node
\draw (355.33,994.75) node [anchor=north west][inner sep=0.75pt]  [font=\footnotesize] [align=left] {{\large {\fontfamily{pcr}\selectfont $f$}$ $}};
% Text Node
\draw (506.18,966.03) node [anchor=north west][inner sep=0.75pt]   [align=left] {{\fontfamily{ptm}\selectfont {\small Worker 1}}};
% Text Node
\draw (506.18,1009.9) node [anchor=north west][inner sep=0.75pt]   [align=left] {{\fontfamily{ptm}\selectfont {\small Worker 2}}};
% Text Node
\draw (506,1051) node [anchor=north west][inner sep=0.75pt]   [align=left] {{\fontfamily{ptm}\selectfont {\small Worker 3}}};
% Text Node
\draw (477.15,935.31) node [anchor=north west][inner sep=0.75pt]   [align=left] {{\fontfamily{ptm}\selectfont Service provider}};
% Text Node
\draw (450.07,1079.76) node [anchor=north west][inner sep=0.75pt]   [align=left] {\begin{minipage}[lt]{39.9pt}\setlength\topsep{0pt}
\begin{center}
{\fontfamily{ptm}\selectfont Existing }\\{\fontfamily{ptm}\selectfont privacy }\\{\fontfamily{ptm}\selectfont barrier}
\end{center}

\end{minipage}};
% Text Node
\draw (367.93,1079.31) node [anchor=north west][inner sep=0.75pt]   [align=left] {\begin{minipage}[lt]{42.61pt}\setlength\topsep{0pt}
\begin{center}
{\fontfamily{ptm}\selectfont Proposed}\\{\fontfamily{ptm}\selectfont privacy }\\{\fontfamily{ptm}\selectfont barrier}
\end{center}

\end{minipage}};
% Text Node
\draw (297.16,1067.62) node [anchor=north west][inner sep=0.75pt]  [font=\large] [align=left] {$\displaystyle \mathcal{A}$};
% Text Node
\draw (262,1068.24) node [anchor=north west][inner sep=0.75pt]  [font=\large] [align=left] {$\displaystyle \boldk$};
% Text Node

\draw (267.24,1100.53) node [anchor=north west][inner sep=0.75pt]  [font=\small] [align=left] {{\large {\fontfamily{pcr}\selectfont $f(\boldx)$}$ $}};

\end{tikzpicture}
 \ \\ \ \\ 
\caption{Shifting the privacy barrier in coded computing.}
    \label{fig:barrier}
\end{figure}
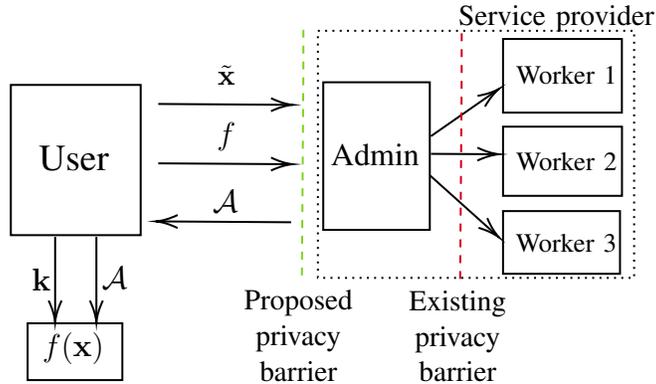

% \begin{figure}[htbp]
%   \centering
%   \includegraphics[width=0.5\textwidth]{barrier.png}
%   % where an .eps filename suffix will be assumed under latex,
%   % and a .pdf suffix will be assumed for pdflatex
%   \caption{Shifting the privacy barrier in coded computation.}
%   \label{fig:barrier}
% \end{figure}

%\red{[New privacy notion and why it's important.]} 
We aim to provide information-theoretic privacy guarantees. However, bounding
%$I(X;\Tilde{X})$ 
the mutual information between the actual and privatized data provides no understanding about which parts of the data are revealed, and to what extent. To remedy that, we leverage the recently developed notion of~\textit{perfect subset privacy}~\cite{raviv2022perfect}, which guarantees zero leakage from~\textit{all} subsets of the data up to a certain size~$r$. This notion provides a uniform privacy metric that is easy to interpret, and~$r$ serves as a tunable privacy parameter that trades better privacy for higher storage and computation requirements, and vice versa.

%as the main privacy metric, which is preferable to other more conventional measures, as delineated next. 
%\red{[Our contributions.]} 
We first generalize the encoding scheme from~\cite{raviv2022perfect} to apply to finite fields of any size.
%so that it works without any assumption on finite field size. 
Next, by utilizing the known notion of \textit{information sets} in Reed-Muller decoding, we devise a scheme which allows the admin to compute~$f$ over the encoded data~$\Tilde{\boldx}$, either locally or in a system without straggling workers; informally, information sets are specific sets of positions in codewords with which one may determine the identity of entire codewords, and they are used specifically to alleviate the download cost of our scheme. Then, in order to adapt the above scheme to tolerate straggling workers, we study a notion that we refer to as \textit{information super-sets} for Reed-Muller codes.
% \red{[Say a few more informal words about what information sets and information supersets are.]}
Information super-sets for a code are multisets whose subsets of a certain size all contain information sets for the same code, and they are used in this work to account for stragglers.
This notion may be of independent interest, and can be seen as related to the problem of puncturing Reed-Muller codes efficiently~\cite{guruswami2017efficiently}. We proceed by devising good straggler-resilient schemes by constructing information super-sets for binary and general Reed-Muller codes. % \red{in the appendix?}.
%After the download cost is reduced, we then minimize the number of workers needed to produce the necessary computation results in the presence of a limited number of stragglers, using a newly established notion of \textit{information super-sets}.

%\red{[Roadmap.]} 
The rest of the paper is structured as follows. The problem setup and computation model are given in Section~\ref{section:preliminaries}. An encoding procedure that achieves perfect subset privacy is provided in Section~\ref{section:rsubsetprivacy}. In Section~\ref{section:reducingD} we reduce the user's download cost in straggler-free systems. In Section~\ref{section:minimizingN} we devise straggler-resilient protocols using information super-sets, which are then constructed for the binary case in Section~\ref{section:binaryRM} and for the general case in Section~\ref{section:generalsupersets}. %which are very commonly used in wireless communication applications.

\section{Problem statement}\label{section:preliminaries}

% Accepted papers will be published in full. For papers that are eligible for the student paper award, please do not forget to remove the comment regarding the paper's eligibility from the abstract.

% The deadline for the submission of the final camera-ready paper is \textbf{May 6, 2024}.  Accepted papers not submitted by that
% date will not appear in the ISIT proceedings and will not be included
% in the technical program of the ISIT.

\subsection{The computation model}\label{section:computationmodel}
Our distributed computation model consists of a storage phase and a computation phase. In what follows, the user holds uniformly distributed data~$\boldx=(x_1, \ldots, x_n) \in \bF_q^n$ chosen from $X = \textrm{Unif}(\bF_{q}^n)$ and encodes the data into a privatized version prior to sharing it. Specifically, the user encodes~$\boldx$ to $\Tilde{\boldx} \triangleq \Tilde{E}(\boldx, \boldk)$, where~$\Tilde{E}$ is some encoding function and~$\boldk \in \bF_q^m$, chosen from~$K = \textrm{Unif}(\bF_{q}^m)$, is some random key, with~$m \ll n$. Then, the user sends $\Tilde{\boldx}$ to the admin, who further encodes $\Tilde{\boldx}$ to $\hat{\boldx} \triangleq \hat{E}(\Tilde{\boldx}) = (\hat{\boldx}_1, \ldots, \hat{\boldx}_N),$ where~$\hat{E}$ is another encoding function. The admin distributes $\hat{\boldx}$ to~$N$ workers, with worker~$i \in [N]$ (where~$[N]\triangleq\{1,2,\ldots,N\}$) storing $\hat{\boldx}_i$. The user wishes to discard~$\boldx$ and only keep~$\boldk$ as side information, which requires less storage. We call this the \textit{storage phase} of the model. %, as illustrated in Fig.~\ref{fig:storage}.

At a later point, the user is interested in computing~$f(\boldx)$, where~$f: \bF_q^n \rightarrow \bF_q$ is a multivariate polynomial of total degree at most~$d$ and individual degree at most~$q-1$ (i.e., $d < n(q-1)$). The polynomial~$f$ is not known during the storage phase.
%Here the degree of $f$ is defined as the maximum of all degrees of the monomials of $f$. Since the individual degree in each term cannot be greater than $q-1$, it follows that $d < n(q-1)$. 
The user sends~$f$ to the admin, who then sends it to the~$N$ workers. For~$i\in[N]$, worker~$i$ performs some computation over~$\hat{\boldx}_i$, and returns the results to the admin. The admin aggregates the results into some information~$\cA$ that is sent back to the user; $\cA$ along with side information~$\boldk$ should suffice for decoding~$f(\boldx)$. We call this the \textit{computation phase} of the model.
Among those~$N$ workers, however, $S$ ones may straggle, i.e., not return a response in a timely manner, and the admin must be able to provide the user with sufficient information to decode~$f(\boldx)$ even in the absence of the computation results from the stragglers.

% \begin{figure}[htbp]
%   \centering
%   \includegraphics[width=0.4\textwidth]{computation.jpg}
%   % where an .eps filename suffix will be assumed under latex,
%   % and a .pdf suffix will be assumed for pdflatex
%   \caption{The computation phase.}
%   \label{fig:computation}
% \end{figure}

\begin{remark}
    % \red{[The connection to the system setting in LCC. Verify, and perhaps add some parameter comparison in the discussion section.]} 
    For simplicity, we assume the user is interested in computing~$f$ on a \emph{single}~$\boldx$, unlike the common system assumption of multiple~$\boldx$'s~\cite{yu2019lagrange}. Our setting can be generalized to multiple~$\boldx$'s (e.g., a dataset with multiple datapoints, or a matrix with many columns) by straightforward repetition.
\end{remark}

% there exist a set~$\Gamma$ of at most $S$ \textit{stragglers}, i.e., workers that are much slower than others. The identities of the stragglers are not known a priori. The admin ends up receiving a set of results $$\cP \triangleq \{f(\hat{X}_i)|i \in [N] \setminus \Gamma, |\Gamma| \le S\}$$ from the ``non-stragglers'' in the system, where $\Gamma \subseteq [N]$ denotes the set of stragglers. Based on $\cP$, the admin prepares and sends to the user an answer $\cA$, from which the user should be able to retrieve $f(X)$. Finally, the user leverages a decoding function $\cD$ along with the random key to obtain the desired result $f(X) = \cD(\cA, K)$. We refer to this as the computation phase. \

\subsection{Perfect subset privacy}
Let~$\Tilde{X}$ be the random variable which is induced by the data distribution~$X$, the distribution of the key~$K$, and the function~$\Tilde{E}$, i.e., $\Tilde{X}=\Tilde{E}(X,K)$. An encoding scheme is called \textit{perfectly private}~\cite{raviv2022perfect} if the transformed and original data have zero mutual information, i.e., $I(X; \Tilde{X}) = 0$. Perfect privacy is a very strong notion, which requires a key as big as the data itself~\cite{CovThom06}; our scheme would trivialize under such requirements, as no storage saving would be possible. Therefore, we must resort to limited information leakage.

However, guarantees of the form~$I(X;\Tilde{X})\le \epsilon$ are rather vague, since they provide no insight about which part of the data is revealed. 
%Perfect privacy is of no use for data sharing though, since the data processing inequality implies that privatized data cannot be used for inference on the original one. 
% On the other hand, to give guarantees related to the notion of \textit{differential privacy} \cite{cuff2016differential} does not imply which part of the data is revealed and to what extent, hence potentially jeopardizing the privacy of individuals in the data. 
To this end, we leverage the notion of \textit{perfect subset privacy}, first developed in \cite{raviv2022perfect} by generalizing \textit{perfect sample privacy} of~\cite{rassouli2019data}. To define this notion, let $\binom{[n]}{r}$ be the collection of all subsets of $[n]$ of size $r$. For any $\cR \subseteq [n]$, we use the notation $X_{\cR} \triangleq (X_i)_{i \in \cR}$, and later extend this notation to matrices and their columns in a natural way.

\begin{definition}\cite[Def.~1]{raviv2022perfect}
    For a privacy parameter~$r \in [n]$, an encoding procedure is said to satisfy \emph{$r$-subset privacy} if $I(X_{\cR}; \Tilde{X}) = 0$ for all $\cR \in \binom{[n]}{r}$.  
\end{definition}

This notion of~$r$-subset privacy provides uniform protection for all $r$-subsets of the data with a tunable parameter~$r$. It is easy to show~\cite{rassouli2019data,raviv2022perfect} that $r$-perfect subset privacy implies that~$I(X;\Tilde{X})\le (n-r)\log_2 q$ bits. Therefore, it can be seen as a special case of bounding information leakage, which has additional uniform and clear guarantees regarding the structure of the privacy protection. It is to be noted that perfect subset privacy and the popular notion of \textit{differential privacy}~\cite{cuff2016differential} are two different privacy paradigms, and that neither implies the other, as explained in~\cite{raviv2022perfect}. 
We also remark that perfect subset privacy is different from collusion resistance~\cite{yu2019lagrange, catalano2005multiparty}, which requires~$I(X;\Tilde{X}_{\cR})=0$ for all $\cR \in \binom{[n]}{r}$. Furthermore, perfect subset privacy is to be distinguished from the notion of~\textit{individual security}~\cite{tarnopolsky2024coding, cohen2018secure, kobayashi2013secure, bhattad2005weakly}, which states that $I(\Tilde{X}_{\cR}; X_{\cQ})=0$ for all $\cR \in \binom{[n]}{r}$ and $\cQ \in \binom{[n]}{n-r}$. Neither of these notions is relevant for our problem definition, since in our setting the service provider has access to~$\tilde{X}$ in its entirety.
%\red{[A comment about the difference between PSS and collusion resistant storage goes here. Do our best to understand individual secrecy, and make a similar comparison here.]}
\subsection{Performance metrics}
We say that any scheme that realizes the computation model in Section~\ref{section:computationmodel} is an~\textit{$(n, q, r, d, S)$-scheme}, where~$n$ is the length of the data~$\boldx$ held by the user, $q$ is the field size, $r$ is the privacy parameter, $d$ is the maximum total degree of~$f$, and~$S$ is the maximum number of stragglers in the system. Below are the relevant metrics by which we judge the merit of such schemes. 

\begin{itemize}
    \item \textbf{The number of workers}, measured by the parameter~$N$.
    \item \textbf{Download cost}, measured by the number~$D$ of~$\bF_q$ symbols the user needs to download in order to complete the decoding process.
    \item \textbf{Side information}, measured by~$m$, the size of the random key~$\boldk$. As mentioned before, we would like~$m$ to be much smaller than~$n$, given the limited storage capability of the user.
\end{itemize}

We would like to minimize all three quantities, and indeed, Section~\ref{section:rsubsetprivacy} deals with minimizing the amount of side information held by the user, Section~\ref{section:reducingD} addresses the problem of reducing download cost, and  Section~\ref{section:minimizingN}, Section~\ref{section:binaryRM}, and Section~\ref{section:generalsupersets} optimize the number of workers needed for computation. %Additionally, we would also like to m
%Minimizing the upload cost of the user (i.e., the cost of sending $\Tilde{X}$ and~$f$) is a source coding problem, that is beyond the scope of this work \red{[Appendix?]}.

\subsection{Our contributions}

We generalize the encoding scheme from~\cite{raviv2022perfect} for an encoding scheme that achieves perfect subset privacy for finite fields of any size. We reduce download cost through evaluations at an information set for a certain Reed-Muller code. Information sets are particular sets of codeword symbols that can be used to identify full codewords. To minimize the number of workers
needed, we first leverage Lagrange Coded Computing (LCC)~\cite{yu2019lagrange} to obtain a scheme that requires large finite fields. Then, we establish the notion of~\textit{information super-sets} for a code, which are multisets whose subsets of a given size all contain information sets for the same code, and present a scheme with no field size restriction. 

The paper then investigates how to construct information super-sets for binary and general Reed-Muller codes, an interesting open problem in and of itself. We come up with constructions of information super-sets for special cases of binary Reed-Muller codes, with which we construct information super-sets for general binary Reed-Muller codes using a recursive approach. We then generalize the above and recursively construct information super-sets for Reed-Muller codes over arbitrary alphabets based on constructions for certain special cases.

We summarize our contributions as such:
\begin{itemize}
    \item We provide an encoding procedure using a generator matrix of any linear code with a certain dual minimum distance property that achieves perfect subset privacy with no field size restriction.
    \item We present a scheme based on information sets for Reed-Muller codes that reduces download cost. The resulting download cost is equal to the dimension of some Reed-Muller code, and is significantly smaller than that of trivial solutions, as discussed next in detail.
    \item We present a scheme based on LCC that, given a certain restriction on field size, reduces the total number of workers needed for straggler resilience in most cases compared to repetition.
    \item We present a scheme based on information super-sets for Reed-Muller codes that reduces the total number of workers needed for straggler resilience to the size of such information super-sets and improves upon repetition and LCC-based schemes.
    \item We study the construction of information super-sets for first-order Reed-Muller codes when there are at most two stragglers, and then, using special case constructions, develop recursive constructions for general information super-sets for Reed-Muller codes.

   % {\color{red} d=1 may not be clear at this point. How about first-order Reed-Muller code instead? }
    % \item We study the construction of information super-sets for Reed-Muller codes when~$d=1$ and there are at most two stragglers.
    % \item Using constructions of information super-sets in special cases, we come up with recursive constructions of general information super-sets for Reed-Muller codes.
\end{itemize}
\section{Obtaining~$r$-subset privacy}\label{section:rsubsetprivacy}

% The paper (A4 or letter size, double-column format, not exceeding 5 pages plus an optional 6th page only containing references) should be formatted as shown in this sample \LaTeX{} file
% \cite{Laport:LaTeX, GMS:LaTeXComp, oetiker_latex, typesetmoser}.

% The use of Microsoft Word or other text processing systems instead of \LaTeX{} is strongly
% discouraged. Users of such systems should attempt to duplicate the
% style of this example, in particular the sizes and type of font, as
% closely as possible.
In this section, we present an encoding procedure which achieves perfect subset privacy. In \cite{raviv2022perfect}, Raviv et al. presented a method that obtains $r$-subset privacy when the finite field is sufficiently large. We first give a brief description of their method and then generalize it to finite fields of any size. 

For distinct elements $w_1, \ldots, w_n \in \bF_q$, let %\begin{align}\label{equation:kv}
   $\bold{\Tilde{x}} \triangleq \boldx + \boldk V$,
%\end{align} 
where~$\boldk \sim \textrm{Unif}(\bF_{q}^r)$ and~$V = V(w_1, \ldots, w_n)$ is an~$r \times n$ Vandermonde matrix over~$w_1, \ldots, w_n$. This method has been proven \cite[Thm.~1]{raviv2022perfect} to achieve~$r$-subset privacy, but it requires~$q \ge n$. 

We now generalize this method to obtain a procedure that guarantees subset privacy without any field size restriction. Specifically, let \begin{align}\label{equation:kg}
    \bold{\Tilde{x}} \triangleq \boldx + \boldk G,
\end{align} where~$\boldk \sim \textrm{Unif}(\bF_{q}^m)$ and~$G \in \bF_q^{m \times n}$ is a generator matrix of any~$[n, m]_q$ linear code $\cC$ that satisfies $d_{\textrm{min}}(\cC^{\perp}) \ge r+1$, where~$d_{\textrm{min}}(\cdot)$ denotes minimum Hamming distance and~$\cC^{	\perp}$ is the dual code of~$\cC$. This implies that for all~$\cR \in \binom{[n]}{r}$, $G_\cR \in \bF_q^{m \times r}$ has full rank. 

\begin{remark}
    Notice that $m \ge r$ according to the Singleton bound, and the equality is achieved with Maximum Distance Separable (MDS) codes. For example, the scheme in \cite{raviv2022perfect} achieves equality with Reed-Solomon codes whose generator matrices are Vandermonde matrices. 
\end{remark}

\begin{remark}
    In our encoding scheme, there is an inherent tradeoff between field size and side information. For example, say $n = 16$ and $r=3$, which implies that $d_{\textrm{min}}(\cC^{\perp}) \ge r + 1 = 4$. If $q=2$, then we need $m=5$ \cite{Grassl:codetables}, whereas if we consider the field $GF(8)$, then we only need $m=4$. The minimum amount of side information needed is $m=3$, achieved by Reed-Solomon codes when the field size is $16$. 
\end{remark}

The following theorem establishes that this generalized encoding indeed achieves perfect subset privacy. 

\begin{theorem} \label{thm:privacy}
    For any~$r \in [n]$, the encoding procedure in~\eqref{equation:kg} satisfies $r$-subset privacy, i.e., $I(X_{\cR}; \Tilde{X}) = 0$ for all $\cR \in \binom{[n]}{r}$. %, where $\Tilde{X}$ is as defined in \eqref{equation:kg}. 
\end{theorem}
\begin{proof}
    For any~$\cR\in\binom{[n]}{r}$, by the law of total probability, we have 
    \begin{align*}
        \Pr&(X_\cR = \boldx_\cR|\Tilde{X}=\Tilde{\boldx}) = \sum_{\theta \in \bF_q^m} \Pr(K=\theta)\Pr(X_\cR = \boldx_\cR|\Tilde{X}=\Tilde{\boldx}, K = \mathbf{\theta})
    \end{align*}
    for every~$\boldx_\cR\in\bF_q^r$ and every~$\Tilde{\boldx}\in\bF_q^n$. Since $K \sim \textrm{Unif}(\bF_{q}^m)$, it follows that $\Pr(K=\theta) = 1/q^m$ for all $\theta \in \bF_q^m$. For any~$\cR\in\binom{[n]}{R}$, any~$\boldx_\cR\in\bF_q^r$, and any~$\Tilde{\boldx}\in\bF_q^n$, let 
    $$\cM_{\cR,\Tilde{\boldx},\boldx_\cR}=\cM \triangleq \{\theta \in \bF_q^m: \Tilde{\boldx}_\cR = \boldx_\cR + \theta G_\cR\}.$$
    The number of solutions $\theta \in \bF_q^m$ for $\theta G_\cR = \Tilde{\boldx}_\cR - \boldx_\cR$ is identical to the number of solutions for $\theta G_\cR = 0$ because of an affine transformation. Thus we have $$|\cM| = |\theta \in \bF_q^m: \{\theta G_\cR = 0\}|.$$ Now, $G_\cR \in \bF_q^{m\times r}$ has full rank, which indicates that $$|\{\theta \in \bF_q^m: \theta G_\cR = 0\}| = q^{m-r},$$ 
    so we have $|\cM| = q^{m-r}$. Since $\Pr(X_\cR = \boldx_\cR|\Tilde{X}=\Tilde{\boldx}, K = \mathbf{\theta}) = 1$ if $\theta \in \cM$ and $0$ otherwise, it follows that $$\Pr(X_\cR = \boldx_\cR|\Tilde{X}=\Tilde{\boldx}) = \frac{1}{q^m} \cdot q^{m-r} = 1/q^r,$$ which implies that the variable $X_\cR | \Tilde{X}=\Tilde{\boldx}$ is uniform over $\bF_{q}^r$. Since $X_\cR \sim \textrm{Unif}(\bF_{q}^r)$, we have $$I(X_{\cR}; \Tilde{X}) = H(X_\cR)-H(X_\cR|\Tilde{X})=0. \qedhere$$ 
\end{proof}
It follows from \eqref{equation:kg} that $\bold{\Tilde{x}} \in \boldx + \cC$, i.e., that~$\boldx$ and~$\bold{\Tilde{x}}$ belong to the same coset of~$\cC$. Therefore, it is possible to reduce \textit{upload} cost from~$n$ to $n-m$ by communicating to the service provider the identity of the coset of~$\cC$ containing~$\boldx$, e.g., by computing a syndrome, which in turn identifies a shift vector that defines the coset\footnote{Interestingly, recent works about privacy of quantized computation over the reals~\cite{raviv2022information, deng2023private, deng2023approximate} employed a similar strategy of communicating a syndrome for saving upload costs, albeit in a very different problem setting.}. Then, the service provider is able to perform computations based on the syndrome and respond to the user accordingly. Further details are given in Appendix~\ref{section:reducingUploadCost}. % \red{Refer to the appendix for further details.}

%It follows from \eqref{equation:kg} that $\Tilde{X} \in X + \cC$, i.e., that $\Tilde{X}$ belongs to a coset of $\cC$. Therefore, it is possible to save upload cost by communicating to the service provider the identity of the coset of~$\cC$ containing~$\Tilde{X}$, e.g., by computing a syndrome, which in turn identifies a shift vector that defines the coset\footnote{Interestingly, recent works about privacy of quantized computation over the reals~\cite{raviv2022information, deng2023private, deng2023approximate} employed a similar strategy of communicating a syndrome for saving upload costs, albeit in a very different problem setting.}. Then, the service provider is able to perform computations based on the syndrome and respond to the user accordingly.

Now that we have demonstrated an encoding scheme which achieves~$r$-subset privacy for any field size~$q$, we turn our attention to devising a compatible computation scheme.% which minimizes the download cost and the number of workers required to compute the multivariate polynomial over the privatized data.

\section{Lowering download cost in straggler-free systems}\label{section:reducingD}

In this section, we first present a trivial scheme which serves as a benchmark. Then, we show that it is possible to reduce download cost through evaluation of a dedicated polynomial at an information set for Reed-Muller codes. 

The trivial scheme works as follows. The user sends~$\bold{\Tilde{x}}$ as defined in~\eqref{equation:kg} to the admin. The admin encodes~$\bold{\Tilde{x}}$ to~$\{\bold{\Tilde{x}} - \boldt G | \boldt \in \bF_q^m\}$ and distributes these to the~$N=q^m$ workers in the storage phase. Later in the computation phase, the admin receives~$f$ from the user. The admin sends the answer~$\cA \triangleq (f(\bold{\Tilde{x}} - \boldt G) | \boldt \in \bF_q^m)$, and then the user can retrieve the element that corresponds to the correct key~$\boldk \in \bF_q^m$ and get the desired result~$f(\boldx) = f(\bold{\Tilde{x}} - \boldk G)$. 
The download cost~$D$ of the trivial scheme equals~$q^m$. This scheme is termed ``trivial'' since it is equivalent to trying all possible guesses for the key. We will next show that it is possible to reduce download cost significantly, by using \textit{information sets} for Reed-Muller codes. %To this end, we first give the definition of \textit{information sets} for linear codes.

% \begin{definition} \cite{key2006information}
%     Let $\cC$ be an $[n,k]_q$ linear code and let $G \in \bF_q^{k \times n}$ be a generator matrix for $\cC$. Then, the set of coordinate positions corresponding to any set of $k$ linearly independent columns of $G$ is called an \textit{information set} for $\cC$. \red{[Would it be preferable to do it formally? ``Formally,~$\cI\in\binom{[n]}{k}$ is an information set if~$G_\cI$ is of rank~$k$. Intuitively, any codeword can be reconstructed fully by observing its entries on any information set.'']}
% \end{definition}
\begin{definition}\label{def:infoset}\cite{key2006information}
    Let~$\cC$ be an $[n,k]_q$ linear code and let $G \in \bF_q^{k \times n}$ be a generator matrix for $\cC$. Then, a set $\cI\in\binom{[n]}{k}$ is an information set if~$G_\cI \in \bF_q^{k\times k}$ is of rank~$k$.
    %Then, the set of coordinate positions corresponding to any set of $k$ linearly independent columns of $G$ is called an \textit{information set} for $\cC$. \red{[Would it be preferable to do it formally? ``Formally,~$\cI\in\binom{[n]}{k}$ is an information set if~$G_\cI$ is of rank~$k$.'']}
\end{definition}
Intuitively, any codeword can be reconstructed in full by observing its entries on any information set. 
%It immediately follows from definition that the size of an information set for a linear code is equal to the dimension of this code. 
Furthermore, we have the following lemma that will be useful later. The proof follows from the definition of minimum distance of linear codes.

\begin{lemma}\label{lemma:nminusdminplusone}
    Let $\cC$ be an $[n,k]_q$ linear code with minimum distance $d_{\textrm{min}}$ and let $G \in \bF_q^{k \times n}$ be a generator matrix for~$\cC$. Then, every subset of $[n]$ of size $n - d_{\textrm{min}} + 1$ contains an information set for $\cC$.
\end{lemma}

% \red{Define info sets for linear codes and define Reed-muller codes in general. Then info sets of RM. Cite Key. Lemma: Every subset of $[p]$ of size $p-d_{min}+1$ contains an info set. Used later. Lemma: if we have evaluation at info sets, linearly combining them gives us eval at other points. Agree on which info set to use. }

We proceed with the definition of Reed-Muller codes \cite{abbe2020reed}. %For a polynomial $f \in \bF_q[x]$ and a vector $z = (z_1, \ldots, z_m) \in \bF_q^m$, let $\textrm{Eval}_z(f) \triangleq f((z_1, \ldots, z_m))$ be the evaluation of $f$ at $z$, and let $\textrm{Eval}(f) \triangleq (\textrm{Eval}_z(f)| z \in \bF_q^m)$ be the evaluation vector of $f$ whose coordinates are the evaluations of $f$ at all $q^m$ points in $\bF_q^m$. Then, the Reed-Muller code with parameters $d$ and $m$ consists of all the evaluation vectors of polynomials with $m$ variables and degree no larger than $d$.

\begin{definition}\cite[Def.~1]{abbe2020reed}
    For parameters~$d$ and~$m$, the Reed-Muller code over $\bF_q$ is defined as $$RM_q(d, m) \triangleq \{\textrm{Eval}(f)\mid f \in \bF_q[\boldx], \deg(f) \le d\},$$ where $\deg(\cdot)$ denotes the total degree of a polynomial, and~$\textrm{Eval}(f)=(f(\boldz))_{\boldz\in \bF_q^m}$.
\end{definition}

% We make the reasonable assumption that $d<m(q-1)$; Reed-Muller codes are non-trivial only if this is true.
%It is immediate that Reed-Muller codes are linear. 
The coordinates of $RM_q(d,m)$ are indexed using the~$q^m$ elements of $\bF_q^m$, in an arbitrary order. In \cite{key2006information}, an explicit description of information sets for Reed-Muller codes is provided. Specifically, by denoting $\bF_q \triangleq \{\alpha_0, \ldots, \alpha_{q-1}\}$, we have that $$\cI_{d,m} \triangleq \{(\alpha_{i_1}, \ldots, \alpha_{i_m})\mid \textstyle\sum_{k=1}^m i_k \le d, 0 \le i_k \le q - 1\}$$ is an information set for $RM_q(d, m)$. 
%In addition, we have the following lemma based on the definition of information sets for Reed-Muller codes.
In addition, the next lemma follows directly from the definitions of information sets and Reed-Muller codes.
\begin{lemma}\label{lemma:evalatinfoset}
    Let~$h$ be an~$m$-variate polynomial over~$\bF_q$ of total degree~$d$. If evaluations of~$h$ at an information set for~$RM_q(d, m)$ are known, then it is possible to obtain evaluations of~$h$ at all other points in~$\bF_q^m$ by \emph{linearly combining} the evaluations at the information set.
\end{lemma}

We now present our scheme based on polynomial evaluation at information sets for Reed-Muller codes. The scheme assumes that the total degree of~$f$ is less than~$m(q-1)$; that is, $d<m(q-1)$. Otherwise, $RM_q(d,m)$ is not well-defined.

The admin encodes~$\bold{\Tilde{x}}$ as defined in~\eqref{equation:kg} to~$\{\bold{\Tilde{x}} - \boldt G | \boldt \in \cI_{d,m}\}$ and distributes them to the workers. At a later point, the admin receives a multivariate polynomial~$f$ of total degree~$d<m(q-1)$ from the user. Define a new multivariate polynomial~$g: \bF_q^m \rightarrow \bF_q$ as \begin{align}\label{equation:polygdef}
    g(\boldt) \triangleq f(\bold{\Tilde{x}} - \boldt G)
\end{align} for all~$\boldt \in \bF_q^{m}$. It can be seen that the total degree of~$g$ is at most~$d$, and that~$g(\boldk) = f(\boldx)$. The admin obtains \begin{align}\label{equation:answerfrominfoset}
    \cA \triangleq (f(\bold{\Tilde{x}} - \boldt G) | \boldt \in \cI_{d,m})
\end{align} from the workers 
%non-stragglers \red{[stragglers were not mentioned yet]} 
and sends it to the user. Since the user has now received evaluations of~$g$ at an information set for~$RM_q(d, m)$, it follows from Lemma \ref{lemma:evalatinfoset} that the user is able to compute~$g(\boldk) = f(\boldx)$ from~$\cA$ and the key $\boldk$. 

It is straightforward to see that this scheme based on information set evaluations reduces download cost of the user with respect to the trivial scheme, provided that $d<m(q-1)$. To be precise, $D$ now equals the size of the information set, $|\cI_{d,m}|$, which is again equal to the dimension of the code $RM_q(d, m)$, written in the sequel as $\lambda(q,d,m)$. Note that to conclude the computation the user combines those~$D=\lambda(q,d,m)$ elements \textit{linearly}, which is normally much less costly than computing~$f$ locally since~$f$ is a polynomial on~$n$ variables.
%\blue{The computational complexity of computing~$f(\boldx)$ in our scheme is linear in~$\lambda(q,d,m)$, whereas computing~$f(\boldx)$ locally would require~$\operatorname{poly}(n)$ operations.}
\begin{theorem}\label{thm:reducedD}
    Assuming $d<m(q-1)$, the scheme based on evaluations at information sets for Reed-Muller codes requires download cost $D = \lambda(q,d,m)$, where $\lambda(q,d,m)$ is the dimension of $RM_q(d,m)$.
\end{theorem}

It is well known \cite{peterson1972error} that if $q=2$, we have $\lambda(q=2,d,m) = \sum_{i=0}^d\binom{m}{i}$, and that if $d < q$, we have $\lambda(q,d<q,m) = \binom{m+d}{d}$. Hence, we have the following reduced download costs. 

\begin{corollary}\label{cor:sumofmchoosei}
    If $q=2$ and $d<m$, then the scheme based on evaluation at information sets of Reed-Muller codes requires download cost $D = \sum_{i=0}^d\binom{m}{i}$.
%\end{corollary}
%\begin{corollary}\label{cor:mplusdchoosed}
    If $d<\min\{q, m(q-1)\}$, then the scheme based on evaluation at information sets of Reed-Muller codes requires download cost $D = \binom{m+d}{d}$.
\end{corollary}

% \red{[It feels like a summary of the resulting parameters belongs here, which includes also the size of the side information w.r.t to the dataset, and the benefits in comparison to the user doing the computation and storage locally. Also, which~$\cC$ to choose and what are the resulting parameters.]}

One example that demonstrates the benefit of using information sets is taking~$\cC^\bot$ as the~$[n=2^t-1,~k \ge n-\frac{r}{2}t,~d_{\textrm{min}} \ge r+1]$ binary BCH code, where~$t \ge 3$ is an integer and~$r$ is an even positive integer~\cite[Sec. 6.1]{ecc_Lin_Cos}. The number~$m$ of parity symbols in~$\cC^\bot$, which equals the dimension of~$\cC$, is upper bounded by~$\frac{r}{2}\log_2(n+1)$. %, where~$n$ is the block length of the code. %; that is, $m \le \frac{r}{2}\log_2(n+1)$. 
This implies download cost~$D \le \sum_{i=0}^d \binom{\frac{r}{2}\log_2(n+1)}{i} \le (1+\frac{r}{2}\log_2(n+1))^d = \operatorname{poly}(\log (n))$, assuming~$d$ is a constant and~$r=O(\operatorname{poly}(\log (n)))$. This download cost is a significant improvement over the trivial scheme with~$D = 2^m>n$, at the same side information size of $m \le \frac{r}{2}\log_2(n+1)=\operatorname{poly}(\log(n))$. As a numeric example, suppose~$n=2^{14}-1$, $r=10$, and $d=2$, which results in~$m=70$. The download cost is~$D \le 2486$ for our scheme, whereas $D=2^{70}$ for the trivial scheme. %\red{[Consider a numeric example, mention that polylog(n) is in comparison to the *one* element that is the result of the computation.]}

%\red{.....VR.... $(S+1)\lambda$ is an upper bound on S-info superset size; add $(S+1)\lambda$ to Lemma 4? We need to show that our constructions outperform it.  Also, compare $(S+1)\lambda$ against LCC.}

\section{Straggler-resilient frameworks}\label{section:minimizingN}

In this section, we seek to minimize the number of workers needed to compute the answer~$\cA$, in the presence of at most~$S$ stragglers. %If $S$ is the maximum number of stragglers, then 
Clearly, $N=(S+1)D$ workers suffice to ensure that the admin can provide the answer in both schemes described in Section~\ref{section:reducingD}. This can be done by storing $(S+1)$ replicas of every $\Tilde{\boldx} - \boldt G$, where $\boldt \in \bF_q^m$ in the trivial scheme and $\boldt \in \cI_{d,m}$ in the information set based scheme.  
Next, we show that it is possible to compute the answer with fewer workers than this repetition-based approach.
%This section is dedicated to reducing~$N$, the number of workers, while accounting for the presence of~$S$ stragglers. 

First, we use Lagrange Code Computing (LCC)~\cite{yu2019lagrange}, which requires large fields. Then, we establish the notion of \textit{information super-sets}, based on information sets defined previously, and devise a scheme without any field size restriction. 

\subsection{LCC-based scheme} 
%\blue{Recall that in the scheme described in Theorem~\ref{thm:reducedD}, given a polynomial~$f$, the admin employs the workers in order to compute~$\cA=\{ g(\boldt)\}_{\boldt\in\cI_{d,m}}$, where~$g$ is an~$m$-variate polynomial of total degree at most~$d$ as defined in~\eqref{equation:polygdef}, and~$\cI_{d,m}$ is an information set for~$RM_q(d,m)$ of size~$\lambda=\lambda(q,d,m)$. Therefore, the LCC scheme can be applied as-is by the admin, over the dataset~$\cI_{d,m}=\{\boldt_1,\ldots,\boldt_\lambda\}$. Namely, the admin constructs a univariate polynomial~$h$ of degree~$\lambda-1$ which evaluates to~$\boldt_1,\ldots,\boldt_\lambda$ at designated~$\bF_q$ elements~$\omega_1,\ldots,\omega_\lambda\in\bF_q$, evaluates~$h$ at~$N$ distinct evaluation points~$\alpha_1,\ldots,\alpha_N\in\bF_q$ with~$N=(\lambda-1)d+S+1$, and sends these evaluation points to the workers, who apply~$g$ over them and return the results. The admin can then obtain~$\cA$ while tolerating~$S$ stragglers. The LCC scheme, however, requires~$N\le q$.}
Recall that in the scheme described in Theorem~\ref{thm:reducedD}, given a polynomial~$f$, the admin employs the workers in order to compute~$\cA=( f(\tilde{\boldx}-\boldt G) \mid \boldt\in\cI_{d,m})$. 
%, where~$g$ is an~$m$-variate polynomial of total degree at most~$d$ as defined in~\eqref{equation:polygdef}. 
Let $\cI_{d,m}=\{\boldt_1,\ldots,\boldt_\lambda\}$, where $\lambda=\lambda(q,d,m)$. 
The LCC scheme can be applied as-is by the admin over the dataset~$\{\tilde{\boldx}-\boldt_i G \mid i \in [\lambda]\}$. Namely, the admin constructs a univariate polynomial~$h$ of degree~$\lambda-1$ which evaluates to~$\tilde{\boldx}-\boldt_1 G,\ldots, \tilde{\boldx}-\boldt_\lambda G$ at designated~$\bF_q$ elements~$\omega_1,\ldots,\omega_\lambda\in\bF_q$, evaluates~$h$ at~$N$ distinct evaluation points~$\alpha_1,\ldots,\alpha_N\in\bF_q$ with~$N=(\lambda-1)d+S+1$, and stores these evaluations in the workers, who apply~$f$ over them and return the results. The admin can then obtain~$\cA$ while tolerating~$S$ stragglers. The LCC scheme, however, requires~$N\le q$.

% LCC~\cite{yu2019lagrange} is a distributed computation framework that provides resilience against stragglers. %and security against malicious workers. 
%It also provides privacy for user data against possibly colluding workers. Directly applying this method to our problem gives the following result. It only works, however, if the field size~$q$ is at least equal to the number of workers $N$. \red{[It's unclear what is the ``LCC method''. We must provide further details.]}

\begin{theorem}\label{theorem:LCC}
    Using the LCC scheme, an $(n, q, r, d, S)$-scheme exists with $N = (\lambda-1)d + S + 1$ workers, provided that $q \ge N$, where~$\lambda$ is the dimension of~$RM_q(d,m)$. The download cost remains~$D=\lambda$.
\end{theorem}
% \blue{VR: Remove the below sentence in red as LCC needs $q \ge N$ which makes it trivial (i.e., $S=0$).}

% \red{We remark that the LCC scheme is compatible with the trivial scheme in the above section as well, and requires $N = (q^m-1)d + S + 1$ workers in that case. } 
We note that if $d<S$, then $(\lambda-1)d+S+1<(S+1)\lambda$.
Next, we demonstrate another scheme based on the notion of  \textit{information super-sets} for Reed-Muller codes, defined shortly. The scheme 
%attains a smaller~$N$ in comparison to Theorem~\ref{theorem:LCC}\red{.....not necessarily true.....VR.....}, and 
has no field size restriction.

\subsection{The information super-set scheme} 

We first define \textit{information super-sets} for linear codes.

\begin{definition}\label{def:infosuperset}
    Let $\cC$ be an $[n,k, d_{\textrm{min}}]_q$ linear code. Then, for any $0 < S < d_{\textrm{min}}$, a multiset~$\cT \subseteq [n]$ is called an~$S$-information super-set for $\cC$ if every~$(|\cT|-S)$-subset\footnote{We say that~$\cS$ is a~$k$-subset of a multiset~$\cT$ if~$\cS$ is a multiset of size~$k$, and the multiplicity of each element in~$\cS$ is at most its multiplicity in~$\cT$.} of~$\cT$ contains an information set for~$\cC$.
    %all subsets of $\cT$ of size $|\cT|-P$ each contain an information set for $\cC$. 
\end{definition}
Let $g$ be the $m$-variate polynomial of degree at most $d$ defined in \eqref{equation:polygdef}. Suppose $\cT$ is an $S$-information super-set for $RM_q(d,m)$, and the number of workers is equal to $|\cT|$. The admin encodes $\bold{\Tilde{x}}$ to $\{\bold{\Tilde{x}} - \boldt G | \boldt \in  \cT\}$ and distributes them to the workers. From the definition of information super-sets, it follows that even in the presence of~$S$ stragglers, the admin receives the evaluations of~$g$ at an information set for~$RM_q(d,m)$.
It then follows from Lemma \ref{lemma:evalatinfoset} that the admin can compute the answer $\cA$ defined in \eqref{equation:answerfrominfoset}. Let $L(q,d,m,S)$ denote the minimum size of an $S$-information super-set for $RM_q(d,m)$. We have the following result.

\begin{theorem}
    % Assuming $d<m(q-1)$, %the scheme in Section \ref{section:reducingD} based on evaluation at information sets of Reed-Muller codes 
    % an $(n, q, r, d, S)$-scheme 
    % works with $N=L(q,d,m,S)$ workers, provided the number of stragglers is at most $S$. \red{[Unclear formulation. Perhaps ``
    For parameters~$n, q, r, d, S$, if there exists an~$[n,m]_q$ linear code~$\cC$ whose dual code~$\cC^\bot$ has minimum distance at least~$r+1$, and if~$d<m(q-1)$, then there exists an $(n, q, r, d, S)$-scheme with $N=L(q,d,m,S)$ workers.
\end{theorem}

Our information super-set based $(n, q, r, d, S)$-scheme for $d<m(q-1)$ can be described as follows. Let~$\cT$ be an $S$-information super-set for $RM_q(d,m)$ of size~$L(q,d,m,S)$.

\begin{enumerate}[label=\roman*)]
    \item The user encodes~$\boldx$ to $\bold{\Tilde{x}} \triangleq \boldx + \boldk G$ according to \eqref{equation:kg} using some random key $\boldk \sim K = \textrm{Unif}(\bF_{q}^m)$ and sends $\bold{\Tilde{x}}$ to the admin. 
    \item The admin encodes~$\bold{\Tilde{x}}$ to  $\bold{\hat{x}} = (\hat{\boldx}_1, \ldots, \hat{\boldx}_N) = (\bold{\Tilde{x}} - \boldt G | \boldt \in \cT)$ %\red{[Why~$\cI_{d,m}$ and not~$\cT$? Also, a vector cannot be equal to a set; we need to find a different notation that works here.]} 
    and distributes them to $N=L(q,d,m,S)$ workers.%, with worker $i \in [N]$ storing $\hat{X}_i$. 
    \item At a later point, the user sends $f$ to the admin, who then sends it to the workers. The workers apply~$f$ on their data, and send the results back to the admin.
    \item The admin obtains $\{f(\bold{\Tilde{x}} - \boldt G)\}_{\boldt \in \cI}$, where~$\cI \subseteq\cT$ is an information set of~$RM_q(d,m)$,
    %\red{[Why~$\cI_{d,m}$ and not~$\cT$?]} 
    %from the computation of (non-straggling) workers and then 
    and sends $\cA \triangleq (f(\bold{\Tilde{x}} - \boldt G))_{\boldt \in \cI_{d,m}}$  to the user. 
  %  \red{[Subtlety: Information set identity is not known to user, but the server may convert one information set to any other using linear operations.]}
    \item The user linearly combines the values in~$\cA$ to obtain~$g(\boldk)=f(\boldx)$.
    %retrieves $f(X)$ from $\cA$ using the key $K$. \red{[A few more words about how that's done? Interpolating~$g$ using previous lemma and then evaluating it?]}
\end{enumerate} 

Next, in order to understand the value of~$N$ in the above scheme, we construct explicit $S$-information super-sets. First, we present a lower- and upper-bound on the size of a minimum~$S$-information super-set. The proofs of the upcoming lemmas
% , which are given in the appendix,
follow directly from the definition of information super-sets (Definition~\ref{def:infosuperset}), and from Lemma~\ref{lemma:nminusdminplusone}.

\begin{lemma}\label{lemma:lowerbound}
    $L(q,d,m,S) \ge \lambda(q,d,m)+S$, where $\lambda(q,d,m)$ is the dimension of $RM_q(d,m)$.
\end{lemma} 
\begin{proof}
    Let~$\cT$ be an $S$-information super-set for~$RM_q(d,m)$ of size~$L(q,d,m,S)$. Suppose~$$L(q,d,m,S) < \lambda(q,d,m)+S,$$ which means~$L(q,d,m,S) - S < \lambda(q,d,m)$. Then, by Definition~\ref{def:infosuperset}, there is a contradiction. That is, $\cT$ is not an~$S$-information super-set for~$RM_q(d,m)$ because every~$|\cT|-S$ subset of~$\cT$ cannot contain an information set for~$RM_q(d,m)$, according to Definition~\ref{def:infoset}.
\end{proof}
%By Definition~\ref{def:infosuperset} and Lemma \ref{lemma:infoset}, we first have the following upper bound for $L(q, d, m, S)$.

\begin{lemma}\label{lemma:weaklbound}
    Let~$d_{\textrm{min}}$ be the minimum distance of $RM_q(d,m)$. If $0 < S < d_{\textrm{min}}$, then 
        $L(q, d, m, S) \le q^m - d_{\textrm{min}} + S + 1.$    In particular,
    if~$q=2$ then $L(q, d, m, S) \le 2^m - 2^{m-d} + S + 1$. 
    %which implies~$d_{\textrm{min}} = 2^{m-d}$, so we have BELONGS IN THE PROOF.
\end{lemma}
\begin{proof}
    From Lemma~\ref{lemma:nminusdminplusone}, it follows that every subset of~$[q^m]$ of size~$q^m-d_{\textrm{min}}+1$ contains an information set for~$RM_q(d,m)$. Let~$\cT $ be a subset of~$[q^m]$. If~$|\cT| = q^m-d_{\textrm{min}}+S+1$, then every~$|\cT|-S$ subset of~$\cT$  contains an information set for~$RM_q(d,m)$, making~$\cT$ an~$S$-information super-set for~$RM_q(d,m)$. Therefore, the minimum size of an~$S$-information super-set for~$RM_q(d,m)$ is at most~$q^m-d_{\textrm{min}}+S+1$. The second statement simply follows from the fact that the minimum distance of~$RM(d,m)$ is~$2^{m-d}$.
\end{proof}
%\red{Is it better to add $(S+1)\lambda$ as a separate comment? Later our (u,u+v) method tries to beat Lemma 4, and if repetition is part of it, proof will become difficult....VR....}
%We thus have the following result.
\begin{corollary}
    Let~$d_{\textrm{min}}$ be the minimum distance of~$RM_q(d,m)$. If $d<m(q-1)$ and $0 < S < d_{\textrm{min}}$, then the number~$N$ of workers required in the information super-set scheme satisfies~$N \le q^m - d_{\textrm{min}} + S + 1$.
    %workers, where $d_{\textrm{min}}$ is the minimum distance of the Reed-Muller code $RM_q(d,m)$. 
    In particular, if $q=2$ then $N \le 2^m - 2^{m-d} + S + 1$.
    %\red{[This statement is not entirely clear. What does it mean ``it works''? What does it mean that~$N$ is at most something? What do we know about this~$N$?]}
\end{corollary}\label{cor:infosupersetbound}

We also note that $L(q,d,m,S) \le (S+1) \lambda(q,d,m)$, as repeating an information set $(S+1)$ times will give an $S$-information super-set.

% \red{We remark the information super-sets are related to the previously defined notions of computational locality and punctured Reed-Muller codes, and the full details are in~??}.

% \red{[Move to appendix.]}

\begin{remark}
    Notions related to information super-sets have appeared previously in the literature. 
In~\cite{rudow2021locality}, \textit{computational locality} is studied for Reed-Muller codes; we refer readers to Appendix~\ref{section:computationallocality} for details. 
Additionally, some results regarding information super-set sizes are implied by existing works~\cite{guruswami2017efficiently} about punctured Reed-Muller codes (see Appendix~\ref{section:puncturedrm}).
\end{remark}

Due to the importance of computations over~$\bF_2$, in Section~\ref{section:binaryRM} we focus on~$q=2$, and construct several information super-sets which outperform the bound in Lemma~\ref{lemma:weaklbound}. Our results on information super-sets for binary Reed-Muller codes can be summarized as follows. We come up with a construction of a~$1$-information super-set for~$RM_2(1,m)$, which yields~$L(2,1,m,1) = m+2$ if~$m$ is even and~$L(2,1,m,1) \le m+3$ if~$m$ is odd. A construction of a~$2$-information super-set for~$RM_2(1,m)$ of size~$2m+1$ is also presented, and yet through a greedy computer search on a set-theoretic formulation
%of~$2$-information super-set construction, 
we obtain smaller~$2$-information super-sets for~$RM_2(1,m)$. 
Then, using the~$(u,u+v)$-construction of binary Reed-Muller codes~\cite{abbe2020reed}, we come up with a construction of~$S$-information super-sets for general binary Reed-Muller codes~$RM_2(d,m)$ which outperforms Lemma~\ref{lemma:weaklbound} in most cases. 
% In Table~\ref{tab:summary}, we give a summary of various methods used in this work in terms of their restrictions and the number of workers needed; numeric examples are given in Table~\ref{tab:info_sup_set_size} in Section~\ref{section:binaryRM}. 

Later, in Section~\ref{section:generalsupersets}, we deal with non-binary Reed-Muller codes, i.e., $q>2$. We first come up with a construction of a~$1$-information super-set for~$RM_q(1,m)$ that yields~$L(q>2,1,m,1) = m+2$. Then, constructions of~$2$-information super-sets are presented.
%of~$2$-information super-set construction,
Generalizing the~$(u,u+v)$-construction of binary Reed-Muller codes, we come up with a construction of~$S$-information super-sets for general Reed-Muller codes~$RM_q(d,m)$, outperforming direct application of Lemma~\ref{lemma:weaklbound}.
In Table~\ref{tab:summary}, we give a summary of various methods used in this work in terms of their restrictions and the number of workers needed; numeric examples are given in Table~\ref{tab:info_sup_set_size} in Section~\ref{section:binaryRM} and in Table~\ref{tab:gen_info_sup_set_size3} in Section~\ref{section:generalsupersets}.

\begin{table*}[ht!]
    \centering
    \resizebox{5.5in}{!}{
\begin{tabular}{ |c|c|c|c|c|c|}
\hline
 & Repetition & LCC & Lemma~\ref{lemma:weaklbound} & Section~\ref{section:binaryRM} & Section~\ref{section:generalsupersets} \\
\hline
 Restriction & None & $q\ge N$, $d<m(q-1)$  & $d<m(q-1)$ & $q=2$, $d<m$ & $d<m(q-1)$\\
 \hline
$N$ & $(S+1)\lambda$ & $(\lambda-1)d+S+1$
 & $q^m-d_{\textrm{min}}+S+1$ &  Theorem~\ref{thm:elldms} & Table~\ref{tab:gen_info_sup_set_size3}
 % ,~\ref{tab:gen_info_sup_set_size4}, and~\ref{tab:gen_info_sup_set_size5}
   \\
 \hline
\end{tabular}
}
 \caption{Summary of various methods. Here $\lambda=\lambda(q,d,m)$ and $d_{\textrm{min}}$ denotes the minimum distance of~$RM_q(d,m)$.}
    \label{tab:summary}
\end{table*}

\section{Information super-sets for binary Reed-Muller codes}\label{section:binaryRM}

This section focuses on constructing information super-sets for binary Reed-Muller codes ($q=2$). For convenience, we omit the parameter~$q$ in this section.
Our techniques rely on a recursive computation which stems from the famous~$(u,u+v)$-construction of Reed-Muller codes~\cite{abbe2020reed}. We begin with necessary special cases in Section~\ref{section:onemone} and Section~\ref{section:onemtwo}, and then proceed with a recursive computation in Section~\ref{section:dms}.

% \subsection{Bounding $L(1,m,1)$}\label{section:onemone}
\subsection{$1$-information super-sets for~$RM(1,m)$}\label{section:onemone}
Let $M$ be an $m \times 2^m$ matrix comprised of all binary column vectors of length $m$.  Define an $(m+1) \times 2^m$ matrix~$\tilde{G}$ as: $$\tilde{G}= \begin{bmatrix}
    M \\
      \1 
\end{bmatrix},$$
i.e., $\tilde{G}$ is obtained by appending an all-one row to $M$. 
It can be verified that $\tilde{G}$ is a generator matrix for $RM(1,m)$. Let $G$ be the $(m+1) \times 2^m$ matrix obtained by row-reducing $\tilde{G}$ via adding the sum of all rows of $M$ to the all-one row. Then, $G$ is also a generator matrix for $RM(1,m)$. 
The columns of $G$ are indexed by subsets of $[m]$; specifically, the entries of the index set indicate the positions of ones in the column, except the last row. For instance, for~$m=4$, the index set~$\{1,3\}$ corresponds to the column~$[1,0,1,0]^\intercal$ in~$M$.
It can also be observed that all columns of $G$ have odd Hamming weights. Without loss of generality, we assume that the columns of $G$ are in the lexicographical ordering of the index sets. 
% \blue{KD: moved these notations here before Theorem 5.}

Since~$RM(1,m)$ is the extended Hadamard code, whose dimension is~$m+1$, it follows that $\lambda(d=1,m)=m+1$, and the bound in Lemma~\ref{lemma:lowerbound} is $L(d=1,m,S) \ge m+1+S$. This bound is attained with equality in the case where~$d=1$, $S=1$, and~$m$ is even, as stated next. 

\begin{theorem}\label{thm:mplustwothree}
    There exists an explicit construction of a $1$-information super-set for $RM(1,m)$, which yields $L(1,m,1) = m+2$ when $m$ is even and $L(1,m,1) \le m+3$ when $m$ is odd.
\end{theorem}
\begin{proof}
    %We begin by establishing some notations. 
% Let $M$ be an $m \times 2^m$ matrix comprised of all binary column vectors of length $m$.  Define an $(m+1) \times 2^m$ matrix~$\tilde{G}$ as: $$\tilde{G}= \begin{bmatrix}
%     M \\
%       \1 
% \end{bmatrix},$$
% i.e., $\tilde{G}$ is obtained by appending an all-one row to $M$. 
% It can be verified that $\tilde{G}$ is a generator matrix for $RM(1,m)$. Let $G$ be the $(m+1) \times 2^m$ matrix obtained by row-reducing $\tilde{G}$ via adding the sum of all rows of $M$ to the all-one row. Then, $G$ is also a generator matrix for $RM(1,m)$. 
% The columns of $G$ are indexed by subsets of $[m]$; specifically, the entries of the index set indicate the positions of ones in the column, except the last row.
% It can also be observed that all columns of $G$ have odd Hamming weights. Without loss of generality, we assume that the columns of $G$ are in the lexicographical ordering of the index sets. We will use these notations in the later proofs as well.  

To construct a 1-information super-set for~$RM(1,m)$, we first pick the $m+1$ index sets of size at most one (i.e., the empty set and $m$ many $1$-element sets). The sub-matrix of $G$ comprised of these columns is the identity matrix~$I_{m+1}$. 

If $m$ is even, then the column of $G$ indexed by $[m]$ is the all-one column. Consider the $(m+1) \times (m+2)$ matrix $G'$ formed by appending the all-one column to $I_{m+1}$. Removing any one column from $G'$ will result in a sub-matrix of rank~$m+1$. Thus, provided~$m$ is even, we have constructed a $1$-information super-set for $RM(1,m)$ of size $m+2$. 

If $m$ is odd, the column of $G$ indexed by $[m]$ has zero in the bottom row and one everywhere else. So, in addition to the~$m+2$ columns already chosen, we fix any $B\subseteq[m]$ of size two, whose corresponding column has Hamming weight three with one appearing in the last row. In this case, let $G'$ 
be the $(m+1) \times (m+3)$ matrix formed by appending the two columns indexed by~$[m]$ and~$B$ to $I_{m+1}$. It can be verified that removing any two columns from $G'$ will not reduce rank. Thus, provided $m$ is odd, we have constructed a $1$-information super set for $RM(1,m)$ of size $m+3$. 
\end{proof}

% Theorem~\ref{thm:mplustwothree}, which is a significant improvement over Lemma~\ref{lemma:weaklbound}, is proved in the appendix. The proof idea is to choose~$m+1$ columns of the generator matrix so that the corresponding submatrix can be row-reduced to the identity matrix, plus either one or two columns (depending on whether~$m$ is even or odd). This way, every~$m+1$ columns out of all the chosen columns are linearly independent over~$\bF_2$. 

%Theorem~\ref{thm:mplustwothree} is a significant improvement over Lemma \ref{lemma:weaklbound}. Moreover, 

\subsection{$2$-information super-sets for~$RM(1,m)$}\label{section:onemtwo}
Theorem~\ref{thm:mplustwothree} can be easily generalized to construct~$S$-information super-sets for any~$S>1$, which provides a (loose) upper bound on $L(1,m,S)$. In particular, it is an easy exercise to show that the multiset given by a union of~$S$ $1$-information super-sets is an~$S$-information super-set, and therefore, $L(1,m,S) \le S\cdot L(1,m,1)$. 
For~$S=2$, this implies that $L(1,m,2) \le 2L(1,m,1) \le 2(m+3)$, but as we now show, smaller constructions exist.

% By generalizing Theorem \ref{thm:mplustwothree} we get that $L(1,m,2) \le 2L(1,m,1) \le 2(m+3)$. The following theorem, whose proof is in the appendix, demonstrates that this bound can be further tightened.

\begin{theorem}\label{thm:twomplusone}
    There exists an explicit construction of a~$2$-information super-set for~$RM(1,m)$ of size~$2m+1$, i.e., $L(1,m,2) \le 2m+1$. 
\end{theorem}

To prove Theorem~\ref{thm:twomplusone} we require the following simple lemma.
\begin{lemma}\label{lemma:dist3}
    Let~$\cC$ be a linear code generated by~$G=[I|P]$ over~$\bF_2$. If every row of~$P$ is distinct and has Hamming weight at least two, then~$d_{\min}(\cC)\ge 3$.
\end{lemma}
\begin{proof}
    We show that the minimum weight of~$\cC$ is at least~$3$. Clearly, $\bF_2$-summation of~$3$ rows of~$G$ or more has Hamming weight at least~$3$, and any one row of~$G$ has Hamming weight at least~$3$. It remains to show that $\bF_2$-summation of any two rows has Hamming weight at least~$3$. Indeed, this summation will have Hamming weight two in the systematic part, and since the respective rows of~$P$ are distinct, it will have Hamming weight at least~$1$ in the non-systematic part, and hence weight at least~$3$ overall.
\end{proof}

\begin{proof}[Proof of Theorem~\ref{thm:twomplusone}]
%\red{Define~$G$. Index the columns of row-reduced~$G$ using subsets of~$[m]$. }
%Now, to prove Theorem~\ref{thm:twomplusone}, r
% Recall the above definition of the generator matrix $G$ for~$RM(1,m)$. 
% Let $G'=[I_{m+1}\mid P]$ be a $(m+1) \times (2m+1)$ sub-matrix of $G$. If $G'$ is a generator matrix of a linear code with a minimum distance of at least three,  then any collection of $(2m-1)$ columns of $G'$ contains $m+1$ linearly independent columns. This will result in a $2$-information super-set for  $RM(1,m)$ of size $2m+1$. To show that the code generated by~$G'$ has $d_{\min} \ge 3$, by Lemma~\ref{lemma:dist3} it suffices to show that every row of $P$ is distinct and has Hamming weight at least two. 

% To construct a $2$-information super-set of size $2m+1$, we first pick the~$m+1$ index sets of size at most one. The remaining $m$ index sets chosen are the following: $\{1,2\}$, $\{m-1,m\}$, and $\{i,i+2\}$ for all $i \in [m-2]$. Note that every $i \in [m]$ is part of exactly two of these $m$ index sets, and no two index sets have an intersection of size more than one. Moreover, it can easily be checked that since~$-1 \ne 0$ in~$\bF_2$, every row of the matrix~$P$ corresponding to these~$m$ index sets has a Hamming weight of at least two. Therefore, the matrix~$P$ satisfies the properties required by Lemma~\ref{lemma:dist3}.
Recall the above definition of the generator matrix $G$ for~$RM(1,m)$. Let $G'=[I_{m+1}\mid P]$ be a $(m+1) \times (2m+1)$ sub-matrix of $G$. If $G'$ is a generator matrix of a linear code with a minimum distance at least three, then any collection of $(2m-1)$ columns of $G'$ contains $m+1$ linearly independent columns, which means~$G'$ gives a~$2$-information super-set for $RM(1,m)$ of size~$2m+1$. To show that constructing such a sub-matrix~$G'$ is possible, by Lemma~\ref{lemma:dist3} it suffices to show how to choose an~$(m+1) \times m$ sub-matrix~$P$ of~$G$ so that every two rows of~$P$ are distinct and every row of~$P$ has Hamming weight at least two. 

% To construct a $2$-information super-set of size $2m+1$, we first pick the~$m+1$ index sets of size at most one. 
To that end, we select the following~$m$ index sets: $\{1,2\}$, $\{m-1,m\}$, and $\{i,i+2\}$ for all $i \in [m-2]$. Note that every $i \in [m]$ belongs to exactly two of these $m$ index sets, and that no two index sets have an intersection of size more than one. Moreover, it can easily be checked that
% since~$-1 \ne 0$ in~$\bF_q$
rows of the matrix~$P$ corresponding to these~$m$ index sets each have a Hamming weight of at least two. Therefore, the matrix~$P$ satisfies the properties required by Lemma~\ref{lemma:dist3}, and~$G'=[I_{m+1}\mid P]$ gives a~$2$-information super-set of size~$2m+1$.
\end{proof}

% \red{[Explain what's the merit of Theorem~\ref{thm:twomplusone} on top of the greedy, i.e., why do we not only mention the greedy.]}

%This is a tighter bound than the one derived based on Theorem \ref{thm:boundforonemone}. 
A slightly more involved construction (obtained through a greedy computer search) provides a smaller~$2$-information super-set for~$RM(1,m)$, albeit without a closed-form expression as in Theorem~\ref{thm:twomplusone}, and numeric results are given in Figure~\ref{fig:greedyresults}. Notice that Lemma~\ref{lemma:lowerbound} and Theorem~\ref{thm:twomplusone}
%Appendix~\ref{section:onemone}, 
suggest that
\begin{align*}
    2 \le L(1,m,2)-m-1 \le m.
\end{align*}
However, Figure~\ref{fig:greedyresults} suggests a sublinear growth of~$L(1,m,2)-m-1$, and hence it is likely that Theorem~\ref{thm:twomplusone} can be improved substantially.

%. The $y$-coordinate in Figure~\ref{fig:greedyresults}, which denotes an upper bound on $L(1,m,2)-\lambda(q=2,d=1,m) = L(1,m,2)-m-1$, seems to grow sub-linearly with~$m$, suggesting a potential improvement over Theorem~\ref{thm:twomplusone}.

%gives a better upper bound on $L(1,m,2)$. 
%\red{Details of this greedy approach will be given in the appendix.}

%; in Figure~\ref{fig:greedyresults} its effectiveness is illustrated numerically.

\begin{figure}[htbp]
  \centering
  \includegraphics[width=0.5\textwidth]{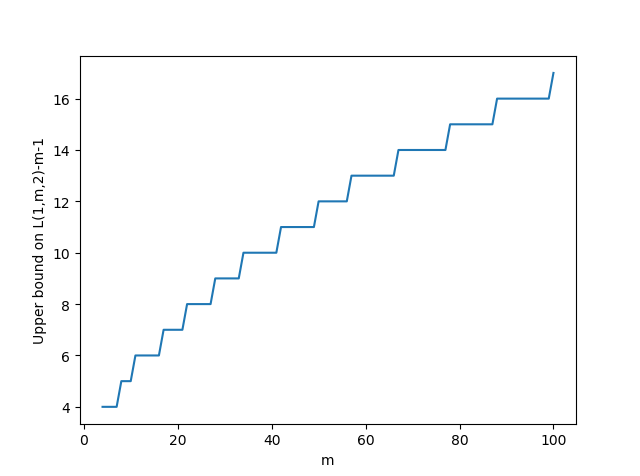}
  % where an .eps filename suffix will be assumed under latex,
  % and a .pdf suffix will be assumed for pdflatex
  \caption{Numeric results ($m \le 100$). %using the greedy method.
  }
  \label{fig:greedyresults}
\end{figure}

%As we can see in the figure, t

%does not grow linearly with $m$ as it does in Theorem \ref{thm:twomplusone}, signaling additional improvement.

\subsection{Greedy approach for~$2$-information super-sets}\label{section:greedysetup}

Recall the definition of the (row-reduced) generator matrix $G$ for~$RM(1,m)$ in the proof of Theorem~\ref{thm:mplustwothree}. Let $G'=[I_{m+1}\mid P]$ be a $(m+1) \times (m+1+u)$ sub-matrix of $G$. To show that the set of indices of $G'$ is a $2$-information super-set, by Lemma~\ref{lemma:dist3}, it suffices to show that every row of $P$ is distinct and has Hamming weight at least two. To construct a $2$-information super-set of size $m+1+u$, we first pick the~$m+1$ index sets of size at most one, which gives the $I_{m+1}$ part of $G'$. Given a collection of subsets $S_1, \dots, S_u \subseteq [m]$ and~$j\in[m]$, define~$T_j\subseteq[u]$ as~$T_j \triangleq \{i \mid j \in S_i \}$, and let $T_0 \triangleq \{j \in [u] \mid |S_j|~\text{is even}\}$. The matrix~$P$ formed by the index sets $S_1, \dots, S_u$ has the required properties if the following conditions are satisfied:  
\begin{itemize}
    \item $|T_i| \ge 2$ for all $i\in[0,m]$ and
    \item $T_i \ne T_j$, for all distinct $i,j\in[0,m]$. 
\end{itemize} 
The goal is to find the smallest $u$ for which such index sets exist. In Appendix~\ref{section:greedydetails} we detail a greedy method which relies upon this set-theoretic formulation of~$2$-information super-set construction, through which we obtained the results in Figure~\ref{fig:greedyresults}.

\subsection{$S$-information super-sets for $RM(d,m)$}\label{section:dms}
%Notice that the assumption in Section \ref{section:reducingD} on the degree of $f$ specializes to $d<m$ in this section.
We turn to construct $S$-information super-sets for any~$S$ using the recursive 
%We turn to bounding~$L(d,m,S)$ in general, where $1<d<m$. We leverage the recursive 
$(u,u+v)$-construction of binary Reed-Muller codes. %~\cite{langton2023generalized}: \red{[Put a much more basic reference here, e.g., Roth's book, preferably an exact reference.]} 
\begin{align*}
    RM&(d,m)= \\ &\{(u,u+v)|u \in RM(d,m-1), v \in RM(d-1,m-1)\}.
\end{align*}

% Using this method, $RM(d,m)$ can be constructed, for example, from $RM(d,m-2)$, $RM(d-1,m-2)$ and $RM(d-2,m-2)$. \red{[Why not only $RM(d,m-1)$ and $RM(d-1,m-1)$?]}

This construction is used to prove the next lemma, which shows that the size of an $S$-information super-set for $RM(d,m)$ is bounded by the sum of the size of an $S$-information super-set for $RM(d,m-1)$ and the size of an $S$-information super-set for $RM(d-1,m-1)$.

\begin{lemma}\label{lemma:sumofsizes}
    $L(d,m,S) \le L(d,m-1,S)+L(d-1,m-1,S)$.
\end{lemma}
%\red{[Proof should probably go here anyway.]}
\begin{proof}
    Let~$G^{(1)}$ be a generator matrix for~$RM(d,m-1)$ and~$G^{(2)}$ a generator matrix for~$RM(d-1,m-1)$. Then, a generator matrix~$G$ for~$RM(d,m)$ can be written as $$G=\begin{bmatrix}
G^{(1)} & G^{(1)}\\
0 & G^{(2)}
\end{bmatrix}.$$  
Let~$\cT_1 \subseteq [2^{m-1}]$ be an
$S$-information super-set for 
$RM(d,m-1)$ and let $\cT_2 \subseteq [2^{m-1}]$ be an
$S$-information super-set for 
$RM(d-1,m-1)$. Define the shifted set $\hat{\cT}_2 \triangleq \cT_2 + 2^{m-1}$, by adding~$2^{m-1}$ to all elements of~$\cT_2$. We claim that $\cT=\cT_1 \cup \hat{\cT}_2$ is an $S$-information super-set for 
$RM(d,m)$; this would conclude the proof since it implies the existence of an $S$-information set for $RM(d,m)$ of size $|\cT_1|+|\cT_2|$.

Indeed, let $\cW$ be any $|\cT|-S$ subset of $\cT$. To prove that $\cT=\cT_1 \cup \hat{\cT}_2$ is an $S$-information super-set for 
$RM(d,m)$, it must be shown that  $G_{\cW}$ has rank $\lambda(d,m)$. 
Note that $\lambda(d,m)=\lambda(d,m-1)+\lambda(d-1,m-1)$. 
Let $\cW_1=\cT_1 \cap \cW$ and $\cW_2=\cT_2 \cap \cW$. It follows that $|\cW_1| \ge |\cT_1|-S$ and $|\cW_2| \ge |\cT_2|-S$. By the definition of information super-sets, there is a $\lambda(d,m-1)$ subset $\cV_1$ of  $\cW_1$ such that $G^{(1)}_{\cV_1}$ is non-singular and a $\lambda(d-1,m-1)$ subset $\cV_2$ of  $\cW_2$ such that $G^{(2)}_{\cV_2}$ is non-singular.  
Consider the following square sub-matrix~$G'$ of~$G_{\cW}$:
$$G'=\begin{bmatrix}
G^{(1)}_{\cV_1} & G^{(1)}_{\cV_2}\\ \ \\ 
0 & G^{(2)}_{\cV_2}
\end{bmatrix}.$$ 
It follows that~$G'$ is invertible (i.e., of rank $\lambda(d,m)$) since~$\det(G')=\det(G^{(1)}_{\cV_1})\cdot \det( G^{(2)}_{\cV_2})$ by a known formula.
\end{proof}

Lemma~\ref{lemma:sumofsizes} implies a recursive method for constructing an~$S$-information super-set for any~$RM(d,m)$, by taking the union of~$S$-information super-sets for~$RM(d,m-1)$ and~$RM(d-1,m-1)$. Also, note that combining Lemma~\ref{lemma:sumofsizes} with Lemma~\ref{lemma:weaklbound} results in
\begin{align} \label{equation:u_uplusv_gain}
    L(d,m,S) &\le (2^{m-1}-2^{m-d-1}+S+1) \nonumber \\
    &\phantom{=}+(2^{m-1}-2^{m-d}+S+1) \nonumber \\
    &= 2^{m}-2^{m-d}-2^{m-d-1}+2S+2,
\end{align}
which is at least as good as Lemma~\ref{lemma:weaklbound} in cases where 
%Direct comparison with the bound in Lemma \ref{lemma:weaklbound} tells us that it is beneficial to apply this technique because it gives a tighter bound on $L(d,m,S)$, provided 
\begin{align}\label{equation:condition}
    S+1 \le 2^{m-d-1}.
\end{align}

Furthermore, the next lemma follows easily from Definition~\ref{def:infosuperset}.

\begin{lemma}\label{lemma:soneplusstwo}
    For all integers $S_1, S_2$ such that $0 < S_1, S_2 < 2^{m-d}$, we have $$L(d,m,S_1+S_2) \le L(d,m,S_1)+L(d,m,S_2).$$
\end{lemma}
\begin{proof}
Let $G$ be a generator matrix for $RM(d,m)$. 
Let~$\cT_1$ be an~$S_1$-information super-set and~$\cT_2$ an~$S_2$-information super-set for~$RM(d,m)$. 
To prove the lemma we show the existence of an $(S_1+S_2)$-information super-set of size $|\cT_1|+|\cT_2|$. 
For this, it suffices to show that removing any $S_1+S_2$ columns from the matrix\footnote{For a matrix~$G$ with~$n$ columns and a multiset~$\cT\subseteq[n]$, we denote by~$G_\cT$ the matrix containing the~$|\cT|$ columns indexed by~$\cT$.} 
$G' \triangleq \begin{bmatrix}
    G_{\cT_1}~G_{\cT_2}
\end{bmatrix}$
results in a matrix of rank $\lambda(d,m)$. It is impossible to simultaneously remove more than $S_1$ columns belonging to $G_{\cT_1}$ and more than $S_2$ columns belonging to $G_{\cT_2}$, while removing at most~$S_1+S_2$ columns overall. As $\cT_1$ and $\cT_2$ are information super-sets, $\lambda(d,m)$ linearly independent columns  belonging to at least one of $G_{\cT_1}$ and  $G_{\cT_2}$ will remain, thereby proving the lemma.
\end{proof}
Lemma~\ref{lemma:soneplusstwo} implies a method for constructing an~$S$-information super-set for any~$RM(d,m)$ by taking the union of~$\floor{S/2}$ $2$-information super-sets for $RM(d,m)$, and potentially one~$1$-information super-set for~$RM(d,m)$ if~$S$ is odd. In particular, we have  $$L(1,i,S)\le  \left\lfloor\frac{S}{2}\right\rfloor (2i+1) + \left(S - 2\left\lfloor\frac{S}{2}\right\rfloor \right) (i+3).$$ Combining this with Lemma~\ref{lemma:weaklbound} gives  
\begin{align} \label{equation:d=1}
  L(1,i,S) \le  U(1,i,S), 
\end{align}
where 
\begin{align*}
U(1,i,S)  \triangleq \min\Big\{&2^{i-1}+S+1,\\
    &\textstyle\left\lfloor\frac{S}{2}\right\rfloor(2i+1) + \left(S - 2\left\lfloor\frac{S}{2}\right\rfloor \right)  (i+3)\Big\}.
\end{align*}

The recursive construction implied by Lemma~\ref{lemma:sumofsizes}, Lemma~\ref{lemma:soneplusstwo}, Section~\ref{section:onemone}, and Section~\ref{section:onemtwo} provides our final algorithm for constructing~$S$-information super-sets for any~$RM(d,m)$, as follows.

%Lemma~\ref{lemma:sumofsizes} and Lemma~\ref{lemma:soneplusstwo} provide an explicit method of constructing an~$S$-information super-set for any~$S\ge1$ and for any~$RM(d,m)$ in a recursive manner. Specifically,

To construct an~$S$-information super-set for~$RM(d,m)$, which is not one of the base cases mentioned shortly, break~$RM(d,m)$ to~$RM(d,m-1)$ and~$RM(d-1,m-1)$, build an~$S$-information super-set for each, and take the union. We consider two base cases which are treated separately. The code~$RM(d,m)$ is considered a base case if either of the following is true.

\begin{enumerate}
    \item The parameters~$d \ge 2$ and~$m>d$ violate~\eqref{equation:condition}, i.e., $S+1>2^{m-d-1}$. In this case, we construct an~$S$-information super-set using Lemma~\ref{lemma:weaklbound}, i.e., we return any set of~$2^m-2^{m-d}+S+1$ coordinates as our~$S$-information super-set.
   %    \red{...VR...this is  incorrect, the base case if the first time $u,u+v$ becomes worse than Lemma 4}
    \item $(d,m)=(1,i)$ for some~$i$. 
    In this case, depending on whichever is smaller, we either apply the method described after Lemma~\ref{lemma:soneplusstwo}, %\red{[It's not technically a recursion, we just return some copies of a~$2$-inf' super-set.]}
    and return the union of the~$\floor{S/2}$ $2$-information super-sets (with the potential additional one $1$-information super-set), or construct an~$S$-information super-set using Lemma~\ref{lemma:weaklbound}.
    
\end{enumerate}

To bound the size of the resulting~$S$-information super-set for the input code~$RM(d,m)$, one may consider the above algorithm as a tree-like structure, and we are left to bound the summation of all sizes of information super-sets that were computed in its leaves. See Figure~\ref{fig:tree} for an example where~$d=4,m=8$ and~$S=2$. In this example, the root of the tree is the code~$RM(4,7)$, the leaves from the first base case are of the form~$RM(i,i+2)$ for~$i=2,3,4$, and the leaf from the second base case is~$RM(1,4)$. We eventually obtain the following general bound on~$L(d,m,S)$ whose proof is given in Appendix~\ref{section:thmsevenproof}.

\begin{figure}
\centering

\tikzset{every picture/.style={line width=0.75pt}} %set default line width to 0.75pt        

\begin{tikzpicture}[x=0.75pt,y=0.75pt,yscale=-1,xscale=1]
%uncomment if require: \path (0,584); %set diagram left start at 0, and has height of 584

%Straight Lines [id:da41402268448396695] 
\draw    (238.28,101.84) -- (183.56,145.14) ;
\draw [shift={(181.99,146.39)}, rotate = 321.65] [color={rgb, 255:red, 0; green, 0; blue, 0 }  ][line width=0.75]    (10.93,-3.29) .. controls (6.95,-1.4) and (3.31,-0.3) .. (0,0) .. controls (3.31,0.3) and (6.95,1.4) .. (10.93,3.29)   ;
%Straight Lines [id:da7963421456437736] 
\draw    (268.8,102.73) -- (319.97,149.93) ;
\draw [shift={(321.44,151.28)}, rotate = 222.68] [color={rgb, 255:red, 0; green, 0; blue, 0 }  ][line width=0.75]    (10.93,-3.29) .. controls (6.95,-1.4) and (3.31,-0.3) .. (0,0) .. controls (3.31,0.3) and (6.95,1.4) .. (10.93,3.29)   ;
%Straight Lines [id:da5422287424640699] 
\draw    (300.86,178.44) -- (260.26,213.66) ;
\draw [shift={(258.75,214.97)}, rotate = 319.07] [color={rgb, 255:red, 0; green, 0; blue, 0 }  ][line width=0.75]    (10.93,-3.29) .. controls (6.95,-1.4) and (3.31,-0.3) .. (0,0) .. controls (3.31,0.3) and (6.95,1.4) .. (10.93,3.29)   ;
%Straight Lines [id:da8879155636578251] 
\draw    (351.6,179.03) -- (386.62,213.27) ;
\draw [shift={(388.05,214.67)}, rotate = 224.35] [color={rgb, 255:red, 0; green, 0; blue, 0 }  ][line width=0.75]    (10.93,-3.29) .. controls (6.95,-1.4) and (3.31,-0.3) .. (0,0) .. controls (3.31,0.3) and (6.95,1.4) .. (10.93,3.29)   ;
%Straight Lines [id:da9525253594878615] 
\draw    (407.05,238.67) -- (444.88,274.27) ;
\draw [shift={(446.34,275.64)}, rotate = 223.26] [color={rgb, 255:red, 0; green, 0; blue, 0 }  ][line width=0.75]    (10.93,-3.29) .. controls (6.95,-1.4) and (3.31,-0.3) .. (0,0) .. controls (3.31,0.3) and (6.95,1.4) .. (10.93,3.29)   ;
%Straight Lines [id:da1527719502211713] 
\draw    (353.89,237.67) -- (318.11,270.18) ;
\draw [shift={(316.63,271.52)}, rotate = 317.74] [color={rgb, 255:red, 0; green, 0; blue, 0 }  ][line width=0.75]    (10.93,-3.29) .. controls (6.95,-1.4) and (3.31,-0.3) .. (0,0) .. controls (3.31,0.3) and (6.95,1.4) .. (10.93,3.29)   ;

% Text Node
\draw (226.67,84.4) node [anchor=north west][inner sep=0.75pt]  [font=\small] [align=left] {{\large {\fontfamily{pcr}\selectfont $RM(4,7)$}}};
% Text Node
\draw (287.44,156.55) node [anchor=north west][inner sep=0.75pt]  [font=\small] [align=left] {{\large {\fontfamily{pcr}\selectfont $RM(3,6)$}}};
% Text Node
\draw (152.8,153.76) node [anchor=north west][inner sep=0.75pt]  [font=\small] [align=left] {{\large {\fontfamily{pcr}\selectfont $RM(4,6)$}}};
% Text Node
\draw (228.34,217.9) node [anchor=north west][inner sep=0.75pt]  [font=\small] [align=left] {{\large {\fontfamily{pcr}\selectfont $RM(3,5)$}}};
% Text Node
\draw (346.69,218.11) node [anchor=north west][inner sep=0.75pt]  [font=\small] [align=left] {{\large {\fontfamily{pcr}\selectfont $RM(2,5)$}}};
% Text Node
\draw (409.31,274.45) node [anchor=north west][inner sep=0.75pt]  [font=\small] [align=left] {{\large {\fontfamily{pcr}\selectfont $RM(1,4)$}}};
% Text Node
\draw (286.04,273.11) node [anchor=north west][inner sep=0.75pt]  [font=\small] [align=left] {{\large {\fontfamily{pcr}\selectfont $RM(2,4)$}}};

\end{tikzpicture}
 \ \\ \ \\ 
\caption{Example of a tree ($S=2$).}
    \label{fig:tree}
\end{figure}
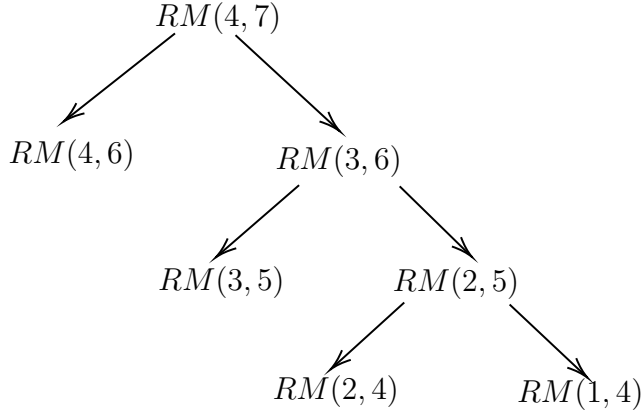

\begin{theorem}\label{thm:elldms}
    If $s\triangleq\log_2(S+1)$ is an integer and $d>1$, then 
    %\red{Doesn't the second case subsume the first? Also, to prevent clutter we can denote~$s=\ceil{\log_2(S+1)}$.}
    \begin{align*}
            L(d,m,S) \le 
        % \sum_{i=2}^d \binom{m-i-s-1}{m-d-s-1} 2^{i+s} + \\ 
        % \sum_{j=2+s}^{m-d+1} \left\lfloor\frac{S}{2}\right\rfloor \cdot L(1,j,2) + \left(S - 2\left\lfloor\frac{S}{2}\right\rfloor \right) \cdot L(1,j,1) \\
         \sum_{i=2}^d \binom{m-i-s-1}{d-i} 2^{i+s}   + \sum_{j=2+s}^{m-d+1} \binom{m-j-1}{d-2} U(1,j,S).
    \end{align*}
    If~$s$ is not an integer and $d>1$, 
    \begin{align*}
        L(d,m,S) 
        % \sum_{i=2}^d \binom{m-i-\left\lceil s\right\rceil-1}{m-d-\left\lceil s\right\rceil-1} \\ 
        % (2^{i+\left\lceil s\right\rceil}-2^{\left\lceil s\right\rceil}+S+1) + \\ 
        % \sum_{j=2+\left\lceil s\right\rceil}^{m-d+1} \left\lfloor\frac{S}{2}\right\rfloor \cdot L(1,j,2) + \left(S - 2\left\lfloor\frac{S}{2}\right\rfloor \right) \cdot L(1,j,1) \\ 
        \le \sum_{i=2}^d \binom{m-i-\left\lceil s\right\rceil-1}{d-i}
        (2^{i+\left\lceil s\right\rceil}-2^{\left\lceil s\right\rceil}+S+1)   + 
\sum_{j=2+\left\lceil s\right\rceil}^{m-d+1} \binom{m-j-1}{d-2} U(1,j,S).
    \end{align*}
    For $d=1$, \eqref{equation:d=1} gives an upper bound on $L(1,m,S)$. 
\end{theorem}

In the above recursive procedure, we break down $RM(d,m)$ only if \eqref{equation:condition} is true. It follows from~\eqref{equation:u_uplusv_gain} that the information super-set size given by Theorem~\ref{thm:elldms} is at least as good as Lemma~\ref{lemma:weaklbound}. In most cases Theorem~\ref{thm:elldms} is better, either due to the gains from the $(u,u+v)$ approach due to~\eqref{equation:u_uplusv_gain} or due to the gains from the method given by Lemma~\ref{lemma:soneplusstwo}.
In Table~\ref{tab:info_sup_set_size}, we provide a few examples that illustrate the merit of using Theorem~\ref{thm:elldms} instead of Lemma~\ref{lemma:weaklbound}. 
\begin{table}[ht!]
    \centering
\begin{tabular}{ |c|c|c|c|c|c| } 
\hline
$d$ & $m$ & $S$ & Theorem~\ref{thm:elldms} &  $2^m - d_{\textrm{min}} + S + 1$ & $(S+1)\lambda(d,m)$ \\
\hline
$1$ & $5$ & $3$ & $19$ & $20$ & $24$ \\ 
\hline
$2$ & $5$ & $1$ & $14$ & $26$ & $32$ \\ 
\hline 
$2$ & $5$ & $2$ & $24$ & $27$ & $48$\\
\hline 
$3$ & $5$ & $1$ & $30$ & $30$ & $52$\\ 
\hline 
$3$ & $6$ & $1$ & $50$ & $58$ & $84$ \\ 
\hline 
$3$ & $6$ & $2$ & $55$ & $59$ & $126$\\ 
\hline 
$3$ & $7$ & $1$ & $78$ & $114$ & $128$ \\ 
\hline
$3$ & $7$ & $2$ & $87$ & $115$ & $192$ \\ 
\hline
$3$ & $7$ & $3$ & $93$ & $116$ & $256$ \\ 
\hline 
$3$ & $8$ & $2$ & $215$ & $227$ & $279$\\ 
\hline 
$4$ & $7$ & $1$ & $112$ & $122$ & $198$\\
\hline 
$4$ & $8$ & $2$ & $208$ & $243$ & $489$\\ 
\hline 
$4$ & $9$ & $2$ & $346$ & $483$ & $768$\\ 
\hline
$4$ & $9$ & $3$ & $407$ & $484$ & $1024$\\ 
\hline 
$4$ & $10$ & $3$ & $677$ & $964$ & $1544$\\
\hline
\end{tabular}
    \caption{Sizes of~$S$-information super-sets for~$RM(d,m)$ for selected parameters.
    % \red{[A little more details about which solutions are considered in each column, with reference if possible.]}
    The methods under comparison are the recursive approach presented throughout Section~\ref{section:binaryRM}, direct application of Lemma~\ref{lemma:weaklbound}, and plain repetition.}
    \label{tab:info_sup_set_size}
\end{table}

It can be seen from Table~\ref{tab:info_sup_set_size} that Theorem~\ref{thm:elldms} compares favorably against  Lemma~\ref{lemma:weaklbound}, and they both outperform the repetition method significantly.

% It can be seen from the table that Theorem~\ref{thm:elldms} compares favorably against  Lemma~\ref{lemma:weaklbound}, and they both outperform the repetition method significantly.

\section{Information super-sets for general Reed-Muller codes}\label{section:generalsupersets}

This section focuses on constructing information super-sets for general Reed-Muller codes ($q>2$). 
We again begin with the special case $d=1$ in Section~\ref{section:generalonesuperset} and~\ref{section:generaltwosuperset}, and then present a recursive procedure in Section~\ref{section:generalqdms} for the general case.

\subsection{$1$-information super-sets for~$RM_q(1,m)$}\label{section:generalonesuperset}

Let~$M$ be an~$m \times q^m$ matrix comprised of all~$q$-ary column vectors of length $m$.  Define an~$(m+1) \times q^m$ matrix~$\tilde{G}$ as: $$\tilde{G}= \begin{bmatrix}
    M \\
      \1 
\end{bmatrix},$$
i.e., $\tilde{G}$ is obtained by appending an all-one row to $M$. 
It can be verified that $\tilde{G}$ is a generator matrix for $RM_q(1,m)$. Let $G$ be the $(m+1) \times q^m$ matrix obtained by row-reducing $\tilde{G}$ via subtracting the sum of all rows of $M$ from the all-one row. Clearly~$G$ is also a generator matrix for $RM_q(1,m)$. 
% \red{VR: $G$ has $q^m$ columns} 
%Also, we once again define \textit{index sets} as subsets of~$[m]$ so that the entries of an index set indicate the positions of ones in a column of~$M$ with only~$1$'s and~$0$'s.
Similar to Section~\ref{section:onemone}, we index the \textit{binary} columns of~$M$ (i.e., the columns which contain only~$0$'s and~$1$'s) using the subsets of~$[m]$ in a natural way.
The columns of~$G$ are indexed similarly.
% It can also be observed that all columns of $G$ have odd Hamming weights. Without loss of generality, we assume that the columns of $G$ are in the lexicographical ordering of the index sets. \blue{KD: moved these notations here before Theorem 5.}

It is well known~\cite{peterson1972error} that $\lambda(q,d=1,m) = m+1$. The following theorem establishes that the bound in Lemma~\ref{lemma:lowerbound} is achieved with equality when~$q>2, d=1$, and~$S=1$.

\begin{theorem}\label{thm:generalmplustwo}
    For any prime power~$q>2$, there exists a construction of a~$1$-information super-set for $RM_q(1,m)$ which yields~$L(q,1,m,1) = m+2$.
\end{theorem}
\begin{proof}
To construct a 1-information super-set for~$RM_q(1,m)$, we first pick the $m+1$ index sets of size at most one (i.e., the empty set and $m$ many $1$-element sets). The sub-matrix of $G$ corresponding to these columns is again the identity matrix~$I_{m+1}$. Then, let~$p$ be the characteristic of~$\bF_q$.

% Since~$q>2$, it follows that the column of~$G$ indexed by~$[m]$ is the all-one column. 
If~$p \nmid m-1$, then we additionally pick the all-one column, which, provided~$q>2$, does not change after row reduction. Consider the $(m+1) \times (m+2)$ matrix $G'$ formed by appending the all-one column to $I_{m+1}$. Removing any one column from $G'$ results in a sub-matrix of rank~$m+1$. 

If~$p \mid m-1$, then the last entry of the all-one column in~$\Tilde{G}$ becomes zero in~$G$ as a result of row reduction. Pick some~$\alpha \in \bF_q \setminus \{0,1\}$. In this case, rather than choosing the all-one column, we select the column in~$\Tilde{G}$ that contains~$\alpha$ in one entry and~$1$ in all other entries. This ensures that the additional column selected does not lose Hamming weight after row reduction. Consequently, the $(m+1) \times (m+2)$ matrix $G'$ formed by appending this extra column to $I_{m+1}$ also has rank~$m+1$ after any one column is removed from it. Thus, provided~$q>2$, we have constructed a~$1$-information super-set for $RM_q(1,m)$ of size $m+2$. 
\end{proof}

% \blue{VR: describe the generator matrix of $RM_q(1,m)$, generalize the index set notion, write proof of Theorem 8 and 9 proofs accordingly}

% The proof works similarly to that of Theorem~\ref{thm:mplustwothree}, except that in this general case we only need one extra parity because the finite field is large enough.

\subsection{$2$-information super-sets for~$RM_q(1,m)$}\label{section:generaltwosuperset}

We investigate how to construct~$2$-information super-sets for~$RM_q(1,m)$ under different parameter regimes. First, we generalize Theorem~\ref{thm:twomplusone} for a construction of~$2$-information super-sets of the same size. Then, we show that under certain field size restrictions, it is possible to obtain smaller information super-sets.

% Theorem~\ref{thm:twomplusone} can be generalized for construction of~$2$-information super-sets for~$RM_q(1,m)$ that have size~$2m+1$ \red{[Give a precise statement and proof. Guide the reader through the flow of ``no~$q$ restriction, strict restriction but small sets, and then relaxed restrictions with larger sets'']}. Next we show that for~$q>2$, it is possible to give an even smaller construction in certain parameter regimes. 

The following lemma generalizes Lemma~\ref{lemma:dist3}, and the proof follows in a similar fashion.

\begin{lemma}\label{lemma:generaldist3}
    Let~$\cC$ be a linear code generated by~$G=[I|P]$ over~$\bF_q$. If no row of~$P$ is a scalar multiple over~$\bF_q$ of another row, and every row of~$P$ has Hamming weight at least two, then~$d_{\min}(\cC)\ge 3$.
\end{lemma}

Theorem~\ref{thm:twomplusone} can be generalized for construction of~$2$-information super-sets for~$RM_q(1,m)$ of size~$2m+1$.

\begin{theorem}\label{thm:generaltwomplusone}
    There exists an explicit construction of a~$2$-information super-set for~$RM_q(1,m)$ of size~$2m+1$, i.e., $L(q,1,m,2) \le 2m+1$. 
\end{theorem}

\begin{proof}
% Let $G$ be a generator matrix for~$RM(1,m)$ and $G'=[I_{m+1}\mid P]$ be a $(m+1) \times (2m+1)$ sub-matrix of $G$. If $G'$ is a generator matrix of a linear code with a minimum distance of at least three,  then any collection of $(2m-1)$ columns of $G'$ contains $m+1$ linearly independent columns. This will result in a $2$-information super-set for  $RM(1,m)$ of size $2m+1$. To show that the code generated by~$G'$ has $d_{\min} \ge 3$, by Lemma~\ref{lemma:dist3} it suffices to show that every row of $P$ is distinct and has Hamming weight at least two. 
Recall the above definition of the generator matrix $G$ for~$RM_q(1,m)$. Let $G'=[I_{m+1}\mid P]$ be a $(m+1) \times (2m+1)$ sub-matrix of $G$. If $G'$ is a generator matrix of a linear code with a minimum distance at least three, then any collection of $(2m-1)$ columns of $G'$ contains $m+1$ linearly independent columns, which means~$G'$ gives a~$2$-information super-set for $RM_q(1,m)$ of size~$2m+1$. To show that constructing such a sub-matrix~$G'$ is possible, by Lemma~\ref{lemma:generaldist3} it suffices to show how to choose an~$(m+1) \times m$ sub-matrix~$P$ of~$G$ so that no row of~$P$ is a scalar multiple of another row and every row of~$P$ has Hamming weight at least two. 

% To construct a $2$-information super-set of size $2m+1$, we first pick the~$m+1$ index sets of size at most one. 
To that end, we select the following~$m$ index sets: $\{1,2\}$, $\{m-1,m\}$, and $\{i,i+2\}$ for all $i \in [m-2]$. Note that every $i \in [m]$ belongs to exactly two of these $m$ index sets, and that no two index sets have an intersection of size more than one. Moreover, it can easily be checked that
% since~$-1 \ne 0$ in~$\bF_q$
rows of the matrix~$P$ corresponding to these~$m$ index sets each have a Hamming weight of at least two. Therefore, the matrix~$P$ satisfies the properties required by Lemma~\ref{lemma:generaldist3}, and~$G'=[I_{m+1}\mid P]$ gives a~$2$-information super-set of size~$2m+1$.
\end{proof}

Next, we establish that smaller~$2$-information super-sets are possible given certain parameter restrictions. First, the following theorem shows a condition under which the lower bound in Lemma~\ref{lemma:lowerbound} is exactly attained.
%\red{VR: Not sure if the present statement is true. %It might be possible to get a better result.} 
\begin{theorem}\label{thm:twoparities}
Let $\bF_q$ be a finite field with characteristic not equal to $2$ or $3$. Then, there exists a construction of a~$2$-information super-set for~$RM_q(1,q-2)$, which yields~$L(q,1,q-2,2) = q+1$.  
%and for~$2)$ any prime power~$q=p^n$ where~$p$ is prime and~$n\ge m$.
\end{theorem}
\begin{proof}
  %  We give such a construction for any prime~$q=m+2$. The other case follows similarly.
  %  Like the proof of Theorem~\ref{thm:generaltwomplusone}, 
  To prove the theorem, we need to find a two-column matrix~$P$ that satisfies the properties of Lemma~\ref{lemma:generaldist3}.
  Let~$x_1, \ldots, x_{q-1}$ be all non-zero elements of~$\bF_q$, with~$x_{q-1} \triangleq 2^{-1}$, and note that~$\sum_{i=1}^{q-1} x_i=0$. 
Let \begin{align*}
        \Tilde{P} \triangleq \begin{bmatrix}
            1 & 1 & \ldots & 1 & 1 \\
            x_1 & x_2 & \ldots & x_{q-2} & 1
        \end{bmatrix}^\intercal,
    \end{align*}
and let~$P$ be the matrix obtained by row-reducing~$\Tilde{P}$ via subtracting the sum of all other rows from the last row of~$\Tilde{P}$. That is,
\begin{align*}
        P \triangleq \begin{bmatrix}
            1 & 1 & \ldots & 1 & 3 \\
            x_1 & x_2 & \ldots & x_{q-2} & 1+2^{-1}
        \end{bmatrix}^\intercal.
    \end{align*}
Let~$p$ be the characteristic of~$\bF_q$. We have~$2^{-1} = (p+1)/2$ and~$-1 = p-1$. Since~$p \ne 3$, it follows that $\frac{p+1}{2} \ne p-1$. This fact, combined with the fact that~$p \ne 2$, gives us that~$1+2^{-1} \ne 0$. Because~$x_1, \ldots, x_{q-2}$ are non-zero, it then follows that all rows of~$P$ have a Hamming weight of two. 

It remains to show that no row of~$P$ is a scalar multiple of another row. Since~$x_1, \ldots, x_{q-2}$ are distinct, showing~$3^{-1}(1+2^{-1}) \notin \{x_1, \ldots, x_{q-2}\}$ suffices. Indeed, we have~$3^{-1}(1+2^{-1}) = 3^{-1}(3\cdot 2^{-1}) = 2^{-1} = x_{q-1}$.
\end{proof}
% \red{VR: proof to be completed...}
%Now, write~$x_{m+1}$ as the remaining nonzero~$\bF_q$ element, i.e., $x_{m+1} \in \bF_q$ and $x_{m+1} \not\in \{0,x_1,\ldots,x_m\}$. In order for the last row of~$P$ not to be a scalar multiple over~$\bF_q$ of any other row of~$P$, we need
%\begin{align*}
    %(1-m)^{-1}(1-\sum_{i=1}^m %x_i)=x_{m+1},
%\end{align*}
%over~$\bF_q$, or equivalently,
%\begin{align*}
 %   (1-m)x_{m+1}+\sum_{i=1}^m x_i=1.
%\end{align*}
%Because~$\sum_{i=1}^{m+1} x_i=0$, it thus suffices to pick~$x_1, \ldots, x_m \neq 0$ such that the remaining~$x_{m+1}$ satisfies~$-m\cdot x_{m+1}=1$, which is indeed possible. 
\begin{remark}
According to the MDS conjecture~\cite{segre1955ovals}, apart from some special cases, $[n,k]_q$ MDS codes exist if and only if~$q \ge n-1$. Because a~$2$-information super-set of size~$m+3$ is equivalent to an~$[m+3,m+1]_q$ MDS code, if the MDS conjecture is true, then such~$2$-information super-sets require field size~$q \ge m+2$, which matches with Theorem~\ref{thm:twoparities}.
\end{remark}
%Given the nature of the above restrictions on field size, we also provide the following construction for less restrictive field sizes~$q$, at the price of larger information super-sets.It is inspired by the construction in Theorem~\ref{thm:twoparities}.
Given the nature of the above restrictions, we also provide the following construction for less restrictive parameter regimes, at the price of larger information super-sets.  

\begin{theorem}\label{thm:paritieswithlambda}
    There exists a construction of a~$2$-information super-set for $RM_q(1,m)$ which yields~$L(q>2,1,m,2) = m+\frac{m}{\gamma}+2$ for any finite field~$\bF_q$ with size at least~$\gamma+2$ and characteristic not equal to~$\gamma-1$, where~$\gamma \mid m$.
\end{theorem}

\begin{proof}
    Let
    \begin{align*}
        \Tilde{P} \triangleq \begin{bmatrix}
            1 & 1 & \ldots & 1 \\
            x_1 & x_2 & \ldots & x_\gamma
        \end{bmatrix}^\intercal,
    \end{align*}
    where~$x_1, \ldots, x_\gamma$ are distinct non-zero elements of~$\bF_q$. For all~$i \in [0, \frac{m}{\gamma}-1]$, define block matrix 
    \begin{align*}
    \Tilde{P}_i \triangleq \begin{bmatrix} \mathbf{0}_{\gamma \times i} & \Tilde{P} & \mathbf{0}_{\gamma \times (\frac{m}{\gamma}-i-1)} \end{bmatrix},
\end{align*} where~$\mathbf{0}$ denotes the zero matrix. Then, define the~$(m+1) \times (\frac{m}{\gamma}+1)$ matrix
    \begin{align*}
        P' \triangleq \begin{bmatrix}
    \Tilde{P}_0 \\
    \Tilde{P}_1 \\
    \vdots \\
    \Tilde{P}_{\frac{m}{\gamma}-1} \\
      \1 
\end{bmatrix},
    \end{align*}
% where~$\Tilde{P_0}$ is the matrix of~$\Tilde{P}$ concatenated with a matrix of zeros from the right, $\Tilde{P_1}$ is the matrix with one leftmost column of zeros followed by~$\Tilde{P}$ concatenated with a matrix of zeros from the right, $\Tilde{P_2}$ is the matrix with two leftmost columns of zeros followed by~$\Tilde{P}$ concatenated with a matrix of zeros from the right, and so on, 
where~$\1$ is the row of all~$1$'s, and let matrix~$P$ be one obtained by row-reducing~$P'$ via subtracting the sum of all rows of~$P'$ from the last row. That the field characteristic is not equal to~$\gamma-1$ ensures that the first entry of the last row of~$P$ is not zero, i.e., ensures that we do not lose Hamming weight from row-reduction, and from the proof of Theorem~\ref{thm:twoparities} we have that if~$q\ge \gamma+2$ then~$P$ has the properties of Lemma~\ref{lemma:generaldist3}.
\end{proof}

\subsection{$S$-information super-sets for~$RM_q(d,m)$}\label{section:generalqdms}

% \red{[Instead of troubling the reader with an equivalent to Theorem~7 we can consider just giving a pseudocode for computing the equivalent, without working out the end formulas. Add a table of values, say, for~$q=3$.]}

As a generalization of Section~\ref{section:dms}, we present a method for constructing~$S$-information super-sets for~$RM_q(d,m)$, assuming still that~$d<m(q-1)$. First, recall that a generator matrix~$G_{d,m}$ for~$RM_q(d,m)$ can be written in a recursive manner. Letting~$w\triangleq \operatorname{min}(d,q-1)$ and~$\bF_q = \{0,\zeta^0,\zeta^1,\ldots,\zeta^{q-2}\}$, we have
\begin{align*}
    G_{d,m} &= \begin{bmatrix}
(\zeta^{q-2})^w G_{d-w,m-1} & \ldots & 0^w G_{d-w,m-1} \\ \vdots & \vdots & \vdots \\ 
(\zeta^{q-2})^0 G_{d,m-1} & \ldots & 0^0 G_{d,m-1}
\end{bmatrix}.
\end{align*}

% a generalization of the~$(u,u+v)$-construction in Section~\ref{section:dms}. 
Then, the next lemma follows from a similar logic to that of Lemma~\ref{lemma:sumofsizes} and also suggests a recursive method for constructing an~$S$-information super-set for any~$RM_q(d,m)$. 

\begin{lemma}\label{lemma:generalsumofsizes}
    $L(q,d,m,S) \le \sum_{i=0}^{w} L(q,d-i,m-1,S)$, where~$w= \operatorname{min}(d,q-1)$.
\end{lemma}

According to~\cite{assmus1998polynomial,romanov2022number,kasami1968new}, the code~$RM_q(d,m)$ has minimum distance equal to~$(q-a)q^{m-b-1}$, where~$d=(q-1)b+a$ and~$0\le a < q-1$. Combining Lemma~\ref{lemma:generalsumofsizes} with Lemma~\ref{lemma:weaklbound} therefore results in
% (\red{where to put $d_{min}$?})
\begin{align}
\label{equation:general_u_uplusv_gain}
    L(q,d,m,S) &\le \sum_{i=0}^{w} [q^{m-1}-(q-a_i)q^{m-b_i-2}+S+1],
\end{align}
where~$a_i=(d-i)\bmod (q-1)$ and~$b_i=\floor{\frac{d-i}{q-i}}$ for every~$i\in\{0,1,\ldots,q-1\}$. This is at least as good as Lemma~\ref{lemma:weaklbound} in cases where
\begin{align}\label{equation:generalcondition}
    \sum_{i=0}^{w} [q^{m-1}-(q-a_i)q^{m-b_i-2}+S+1] \le q^{m}-(q-a_0)q^{m-b_0-1}+S+1.
\end{align}

Furthermore, Lemma~\ref{lemma:soneplusstwo} holds for any~$q$: for all integers $S_1, S_2$ such that $0 < S_1, S_2 < (q-a_0)q^{m-b_0-1}$, we have $L(q,d,m,S_1+S_2) \le L(q,d,m,S_1)+L(q,d,m,S_2)$. Therefore, we formulate our final method for constructing~$S$-information super-sets for any~$RM_q(d,m)$, as follows. Break~$RM_q(d,m)$ into~$\{RM_q(d-i,m-1)\}_{i=0}^{w}$, build an~$S$-information super-set for each, take the union, and continue this breakdown until one of the following base cases is reached:
% We consider two base cases which are treated separately. The code~$RM_q(d,m)$ is considered a base case if either of the following is true.

\begin{enumerate}
    \item The parameters~$d \ge 2$ and~$m(q-1)>d$ violate~\eqref{equation:generalcondition}. In this case, we construct an~$S$-information super-set using Lemma~\ref{lemma:weaklbound}, i.e., we return any set of~$q^{m}-(q-d\bmod(q-1))q^{m-\floor{\frac{d}{q-1}}-1}+S+1$ coordinates as an~$S$-information super-set.

    \item $(d,m)=(1,i)$ for some~$i$. 
    In this case, depending on whichever is smaller, we either apply the method implied by Lemma~\ref{lemma:soneplusstwo}, %\red{[It's not technically a recursion, we just return some copies of a~$2$-inf' super-set.]}
    and return the union of the~$\floor{S/2}$ $2$-information super-sets (with the potential additional one $1$-information super-set), or construct an~$S$-information super-set using Lemma~\ref{lemma:weaklbound}.

    \item $(d,m)=(0,i)$ for some~$i$. In this case, it is easy to check that both sides of~\eqref{equation:generalcondition} equal~$S+1$, so we simply return any set of~$S+1$ coordinates as 
    an~$S$-information super-set.
\end{enumerate}

% The above procedure is summarized in~\textbf{Algorithm~\ref{algo:infosupersets}}. We assume that~$d<m(q-1)$ holds.

% \begin{algorithm}[ht!]
%   \caption{Recursive method for~$S$-information super-sets for~$RM_q(d,m)$}\label{algo:infosupersets}
%   \Input{Positive integers~$q,d,m,S$ such that~$d<m(q-1)$, $S>0$.}
%   \Output{Multiset~$\cT$, an~$S$-information super-set for~$RM_q(d,m)$.}
%   \If{$d=0$}{%
%   $\cT \gets$ any set of~$S+1$ coordinates\;
%   }
%   \If{$d=1$}{%
%   $\cT_1 \gets$ union of~$\floor{S/2}$ $2$-information super-sets for~$RM_q(1,m)$ (with a possible additional one $1$-information super-set)\;
%   $\cT_2 \gets$ $S$-information super-set for~$RM_q(1,m)$ constructed using Lemma~\ref{lemma:weaklbound} only\;
%   \If{$|\cT_1|<|\cT_2|$}{%
%   $\cT \gets \cT_1$\;
%   }
%   \Else{%
%   $\cT \gets \cT_2$\;
%   }
%   }
%   \While{\eqref{equation:generalcondition} holds}{%
%   $w \gets \operatorname{min}(d,q-1)$\;
%   Break~$RM_q(d,m)$ into~$\{RM_q(d-i,m-1)\}_{i=0}^{w}$\;
%   $\cT_0, \cT_1, \ldots, \cT_w \gets S$-information super-set for~$RM_q(d,m-1),RM_q(d-1,m-1),\ldots,RM_q(d-w,m-1)$, respectively\;
%   $\cT \gets \bigcup_{i=0}^w \cT_i$\; 
%   }
%   $\cT \gets$ $S$-information super-set constructed using Lemma~\ref{lemma:weaklbound} only\;
% \end{algorithm}
Deriving a bound such as the one given in Theorem~\ref{thm:elldms} is possible, but highly cumbersome. Nevertheless, the breakdown described above is easily implementable using computer code. Hence, to illustrate the merit of our techniques over using Lemma~\ref{lemma:weaklbound} directly, we provide several examples for~$q=3,4,5$ in Table~\ref{tab:gen_info_sup_set_size3}. It is easily seen that the recursive approach outperforms straightforward application of Lemma~\ref{lemma:weaklbound} quite significantly.

\begin{table}[ht!]
    \centering
\begin{tabular}{ |c|c|c||c|c||c|c||c|c| } 
\hline
\multicolumn{3}{|c||}{Parameters} & \multicolumn{2}{c||}{$q=3$} & \multicolumn{2}{c||}{$q=4$} & \multicolumn{2}{c|}{$q=5$} \\
\hline
$d$ & $m$ & $S$ & Recursion & Lemma~\ref{lemma:weaklbound} & Recursion & Lemma~\ref{lemma:weaklbound} & Recursion & Lemma~\ref{lemma:weaklbound} \\
\hline
$1$ & $5$ & $3$ & $18$ & $85$ & $18$ & $260$ & $18$ & $629$ \\ 
\hline
$2$ & $5$ & $1$ & $29$ & $164$ & $29$ & $514$ & $29$ & $1252$ \\ 
\hline 
$2$ & $5$ & $2$ & $39$ & $165$ & $41$ & $515$ & $41$ & $1253$ \\
\hline 
$3$ & $5$ & $1$ & $67$ & $191$ & $77$ & $770$ & $77$ & $1877$ \\ 
\hline 
$3$ & $6$ & $1$ & $103$ & $569$ & $115$ & $3074$ & $115$ & $9377$ \\ 
\hline 
$3$ & $6$ & $2$ & $134$ & $570$ & $159$ & $3075$ & $161$ & $9378$ \\ 
\hline 
$3$ & $7$ & $1$ & $149$ & $1703$ & $163$ & $12290$ & $163$ & $46877$ \\ 
\hline
$3$ & $7$ & $2$ & $200$ & $1704$ & $230$ & $12291$ & $232$ & $46878$ \\ 
\hline
$3$ & $7$ & $3$ & $283$ & $1705$ & $326$ & $12292$ & $344$ & $46879$ \\ 
\hline 
$3$ & $8$ & $2$ & $284$ & $5106$ & $319$ & $49155$ & $321$ & $234378$ \\ 
\hline 
$4$ & $7$ & $1$ & $359$ & $1946$ & $435$ & $13314$ & $449$ & $62502$ \\
\hline 
$4$ & $8$ & $2$ & $722$ & $5835$ & $907$ & $53251$ & $944$ & $312503$ \\ 
\hline 
$4$ & $9$ & $2$ & $1093$ & $17499$ & $1332$ & $212995$ & $1374$ & $1562503$ \\ 
\hline
$4$ & $9$ & $3$ & $1524$ & $17500$ & $1879$ & $212996$ & $2019$ & $1562504$\\ 
\hline 
$4$ & $10$ & $3$ & $2254$ & $52492$ & $2699$ & $851972$ & $2870$ & $7812504$ \\
\hline
\end{tabular}
    \caption{Sizes of~$S$-information super-sets for~$RM_q(d,m)$ for selected parameters. The methods under comparison are the recursive approach presented throughout Section~\ref{section:generalsupersets} and the direct application of Lemma~\ref{lemma:weaklbound}. It can be seen that there is marginal dependence of the performance of the recursive method on the field size~$q$, while the performance of Lemma~\ref{lemma:weaklbound} depends on~$q$ in a drastic manner.}
    \label{tab:gen_info_sup_set_size3}
\end{table}

\section{Conclusion}
% In this work, we focus on techniques to protect privacy against the service provider as a whole, instead of assuming restricted collusion among workers.
In this work, we reexamine the notion of restricted collusion among worker nodes in a distributed computing system. We suggest that in many common cases the worker nodes are completely free to collude, as they are operating under a single service provider which itself poses a privacy risk. The privacy notion we consider is that of perfect subset privacy. We use information sets of Reed-Muller codes to reduce the download requirement of our schemes. To tackle stragglers, we explore the notion of information super-sets of Reed-Muller codes and present some constructions of information super-sets. 
Information super-sets might be of independent interest since they shed light on the behavior of Reed-Muller codes under puncturing. Characterizing bounds and algorithms for constructing information super-sets remains an interesting problem for future work.
% Information super-sets could be of independent interest for other applications, and characterizing the optimal size of information super-sets remains an interesting open problem.

% To bound the size of the resulting~$S$-information super-set for the input code~$RM(d,m)$, one may consider the above algorithm as a tree-like structure, and we are left to bound the summation of all sizes of information super-sets that were computed in its leaves. See Figure~\ref{fig:tree} for an example where~$d=4,m=8$ and~$S=2$. In this example, the root of the tree is the code~$RM(4,7)$, the leaves from the first base case are of the form~$RM(i,i+2)$ for~$i=2,3,4$, and the leaf from the second base case is~$RM(1,4)$. 
\newpage 
\printbibliography

@article{lee2017speeding,
  title={Speeding up distributed machine learning using codes},
  author={Lee, Kangwook and Lam, Maximilian and Pedarsani, Ramtin and Papailiopoulos, Dimitris and Ramchandran, Kannan},
  journal={IEEE Transactions on Information Theory},
  volume={64},
  number={3},
  pages={1514--1529},
  year={2017},
  publisher={IEEE}
}

@article{kasami1968new,
  title={New generalizations of the Reed-Muller codes--I: Primitive codes},
  author={Kasami, Tadao and Lin, Shu and Peterson, W},
  journal={IEEE Transactions on information theory},
  volume={14},
  number={2},
  pages={189--199},
  year={1968},
  publisher={IEEE}
}

@article{segre1955ovals,
  title={Ovals in a finite projective plane},
  author={Segre, Beniamino},
  journal={Canadian Journal of Mathematics},
  volume={7},
  pages={414--416},
  year={1955},
  publisher={Cambridge University Press}
}

@article{peterson1972error,
  title={Error-correcting codes},
  author={Peterson, WW},
  journal={Cambridge, MA: MIT Press google schola},
  volume={2},
  pages={208--213},
  year={1972}
}

@article{romanov2022number,
  title={On the number of q-ary quasi-perfect codes with covering radius 2},
  author={Romanov, Alexander M},
  journal={Designs, Codes and Cryptography},
  volume={90},
  number={8},
  pages={1713--1719},
  year={2022},
  publisher={Springer}
}

@article{assmus1998polynomial,
  title={Polynomial codes and finite geometries},
  author={Assmus Jr, EF and Key, JD},
  journal={Handbook of coding theory},
  volume={2},
  number={part 2},
  pages={1269--1343},
  year={1998},
  publisher={Citeseer}
}

@book{ecc_Lin_Cos,
    title={{Error Control Coding: Fundamentals and Applications}},
    author={Lin, Shu and Costello, Daniel J.},
    year={2004},
    publisher={Pearson Prentice Hall}
}

@inproceedings{yu2019lagrange,
  title={Lagrange coded computing: Optimal design for resiliency, security, and privacy},
  author={Yu, Qian and Li, Songze and Raviv, Netanel and Kalan, Seyed Mohammadreza Mousavi and Soltanolkotabi, Mahdi and Avestimehr, Salman A},
  booktitle={The 22nd International Conference on Artificial Intelligence and Statistics},
  pages={1215--1225},
  year={2019},
  organization={PMLR}
}

@article{wang2022breaking,
  title={Breaking blockchain’s communication barrier with coded computation},
  author={Wang, Canran and Raviv, Netanel},
  journal={IEEE Journal on Selected Areas in Information Theory},
  volume={3},
  number={2},
  pages={405--421},
  year={2022},
  publisher={IEEE}
}

@inproceedings{raviv2022perfect,
  title={Perfect Subset Privacy for Data Sharing and Learning},
  author={Raviv, Netanel and Goldfeld, Ziv},
  booktitle={2022 IEEE International Symposium on Information Theory (ISIT)},
  pages={1850--1855},
  year={2022},
  organization={IEEE}
}

@inproceedings{cuff2016differential,
  title={Differential privacy as a mutual information constraint},
  author={Cuff, Paul and Yu, Lanqing},
  booktitle={Proceedings of the 2016 ACM SIGSAC Conference on Computer and Communications Security},
  pages={43--54},
  year={2016}
}

@inproceedings{raviv2022information,
  title={Information theoretic private inference in quantized models},
  author={Raviv, Netanel and Bitar, Rawad and Yaakobi, Eitan},
  booktitle={2022 IEEE International Symposium on Information Theory (ISIT)},
  pages={1641--1646},
  year={2022},
  organization={IEEE}
}

@article{deng2023private,
  title={Private Inference in Quantized Models},
  author={Deng, Zirui and Ramkumar, Vinayak and Bitar, Rawad and Raviv, Netanel},
  journal={arXiv preprint arXiv:2311.13686},
  year={2023}
}

@article{key2006information,
  title={Information sets and partial permutation decoding for codes from finite geometries},
  author={Key, Jennifer D and McDonough, TP and Mavron, Vassili C},
  journal={Finite Fields and Their Applications},
  volume={12},
  number={2},
  pages={232--247},
  year={2006},
  publisher={Elsevier}
}

@article{abbe2020reed,
  title={Reed--Muller codes: Theory and algorithms},
  author={Abbe, Emmanuel and Shpilka, Amir and Ye, Min},
  journal={IEEE Transactions on Information Theory},
  volume={67},
  number={6},
  pages={3251--3277},
  year={2020},
  publisher={IEEE}
}

@article{guruswami2017efficiently,
  title={Efficiently list-decodable punctured Reed-Muller codes},
  author={Guruswami, Venkatesan and Jin, Lingfei and Xing, Chaoping},
  journal={IEEE Transactions on Information Theory},
  volume={63},
  number={7},
  pages={4317--4324},
  year={2017},
  publisher={IEEE}
}

@Misc{Grassl:codetables,
  author =       "Grassl, Markus",
  title =        "{Bounds on the minimum distance of linear codes and quantum codes}",
  howpublished = "Online available at \url{http://www.codetables.de}",
  year =         "2007",
  note =         "Accessed on 2024-01-08"
}

@article{rassouli2019data,
  title={Data disclosure under perfect sample privacy},
  author={Rassouli, Borzoo and Rosas, Fernando E and G{\"u}nd{\"u}z, Deniz},
  journal={IEEE Transactions on Information Forensics and Security},
  volume={15},
  pages={2012--2025},
  year={2019},
  publisher={IEEE}
}

@inproceedings{deng2023approximate,
  title={Approximate private inference in quantized models},
  author={Deng, Zirui and Raviv, Netanel},
  booktitle={2023 IEEE International Symposium on Information Theory (ISIT)},
  pages={1597--1602},
  year={2023},
  organization={IEEE}
}

@BOOK{CovThom06,
	title = {Elements of Information Theory},
	publisher = {Wiley},
	year = {2006},
	author = {T. M. Cover and J. A. Thomas},
	address = {New-York},
	edition = {2nd}
}

@article{catalano2005multiparty,
  title={Multiparty computation, an introduction},
  author={Catalano, Dario and Cramer, Ronald and Di Crescenzo, Giovanni and Darmg{\aa}rd, Ivan and Pointcheval, David and Takagi, Tsuyoshi and Cramer, Ronald and Damg{\aa}rd, Ivan},
  journal={Contemporary cryptology},
  pages={41--87},
  year={2005},
  publisher={Springer}
}

@article{tarnopolsky2024coding,
  title={Coding-Based Hybrid Post-Quantum Cryptosystem for Non-Uniform Information},
  author={Tarnopolsky, Saar and Cohen, Alejandro},
  journal={arXiv preprint arXiv:2402.08407},
  year={2024}
}

@article{cohen2018secure,
  title={Secure multi-source multicast},
  author={Cohen, Alejandro and Cohen, Asaf and Medard, Muriel and Gurewitz, Omer},
  journal={IEEE Transactions on Communications},
  volume={67},
  number={1},
  pages={708--723},
  year={2018},
  publisher={IEEE}
}

@article{kobayashi2013secure,
  title={Secure multiplex coding attaining channel capacity in wiretap channels},
  author={Kobayashi, Daisuke and Yamamoto, Hirosuke and Ogawa, Tomohiro},
  journal={IEEE transactions on information theory},
  volume={59},
  number={12},
  pages={8131--8143},
  year={2013},
  publisher={IEEE}
}

@article{bhattad2005weakly,
  title={Weakly secure network coding},
  author={Bhattad, Kapil and Narayanan, Krishna R and others},
  journal={NetCod, Apr},
  volume={104},
  pages={8--20},
  year={2005}
}

@inproceedings{rudow2021locality,
  title={A locality-based lens for coded computation},
  author={Rudow, Michael and Rashmi, KV and Guruswami, Venkatesan},
  booktitle={2021 IEEE International Symposium on Information Theory (ISIT)},
  pages={1070--1075},
  year={2021},
  organization={IEEE}
}
\clearpage
\appendices

\ifarXiv

\section{Reducing Upload Cost}\label{section:reducingUploadCost}
Let $G$ be a generator matrix and $H$ be a parity check matrix of any $[n, m]_q$ linear code $\cC$ that satisfies $d_{\textrm{min}}(\cC^{\perp}) \ge r+1$. 
Let $B(H,\bolds)$ be a publicly known deterministic algorithm to compute a solution $\boldy$ to $\bolds H^\intercal=\boldy$. In other words, given any syndrome, there is a fixed shift vector that defines the corresponding coset. 
The user computes the syndrome $\boldx H^\intercal$ and sends it to the service provider.
Both the user and the service provider can obtain $\bold{\Tilde{x}} \triangleq B(H,\boldx H^\intercal).$ 
Since $\boldx H^\intercal=\bold{\tilde{x}}H^\intercal$, it follows that $\boldx-\bold{\tilde{x}} \in \mathcal{C}$. There exists exactly one $\boldk \in \bF_q^{m}$ such that $\boldk G=\boldx-\bold{\tilde{x}}$. The user stores this $\boldk$ as side information. This scheme is equivalent to \eqref{equation:kg} with key $\boldk$. 
The upload cost is clearly reduced from $n$ to $n-m$. 

% \subsection{Proof of Lemma~\ref{lemma:lowerbound}}

% Let~$\cT$ be an $S$-information super-set for~$RM_q(d,m)$ of size~$L(q,d,m,S)$. Suppose~$L(q,d,m,S) < \lambda(q,d,m)+S$, which means~$L(q,d,m,S) - S < \lambda(q,d,m)$. Then, by Definition~\ref{def:infosuperset}, there is a contradiction. That is, $\cT$ is not an~$S$-information super-set for~$RM_q(d,m)$ because every~$|\cT|-S$ subset of~$\cT$ cannot contain an information set for~$RM_q(d,m)$, according to Definition~\ref{def:infoset}.

% \subsection{Proof of Lemma~\ref{lemma:weaklbound}}

% From Lemma~\ref{lemma:nminusdminplusone}, it follows that every subset of~$[q^m]$ of size~$q^m-d_{\textrm{min}}+1$ contains an information set for~$RM_q(d,m)$. Let~$\cT $ be a subset of~$[q^m]$. If~$|\cT| = q^m-d_{\textrm{min}}+S+1$, then every~$|\cT|-S$ subset of~$\cT$  contains an information set for~$RM_q(d,m)$, making~$\cT$ an~$S$-information super-set for~$RM_q(d,m)$. Therefore, the minimum size of an~$S$-information super-set for~$RM_q(d,m)$ is at most~$q^m-d_{\textrm{min}}+S+1$. The second statement simply follows from the fact that the minimum distance of~$RM(d,m)$ is~$2^{m-d}$.

\section{Computational Locality}\label{section:computationallocality}
%\red{[To add: context (why this is studied?) and references, and mention that the connection to information super-sets was not addressed in that paper.]}
The notion of computational locality was introduced in \cite{rudow2021locality} as a proxy for the minimum number of workers necessary for coded computation in the presence of stragglers. Let $\cC$ be a linear code of block length~$n$, let~$\ell \le n$ and~$S$ be non-negative integers, and let $\cI \in \binom{[n]}{\ell}$. The computational locality~$L_{\cI,S}$ is the minimum size of multiset~$\cT \subseteq [n]$ such that~$G_\cI$ is in the span of any collection of $|\cT|-S$ columns of~$G_{\cT}$. The computational locality~$L_{\ell,S}$ of code~$\cC$ is $\operatorname{max}_{\cI \in \binom{[n]}{\ell}}L_{\cI,S}$. %Under this notation, the minimum size of~$S$-information super-sets is equivalent to the minimum $L_{\cI,S}$ among all information sets~$\cI$. Clearly, the two notions are different. We also note that connection to information super-sets was not addressed in \cite{rudow2021locality}.
In the present paper we wish to construct~$S$-information super-sets of minimum size; it is evident that the size of such~$S$-information super-sets equals to the minimum of~$L_{\cI,S}$ across all information sets~$\cI\subseteq[n]$. We also note that~\cite{rudow2021locality} did not explicitly study or construct information super-sets.

\section{Punctured Reed-Muller Codes}\label{section:puncturedrm}
In a sense, the problem of puncturing Reed-Muller codes is dual to that of constructing information super-sets for our application. In the former, one fixes the block length of the punctured code, and wishes to maximize its minimum distance. In our setting, we fix the minimum distance (implied by the parameter~$S$), and wish to minimize the block length. One relevant work known to the authors is~\cite{guruswami2017efficiently},
%There are several results in the literature about punctured Reed-Muller codes that are useful in this direction. For instance, 
 which provides a tighter bound on $L(q, d, m, S)$ than Lemma~\ref{lemma:weaklbound} for some values of~$S$, assuming certain structures of the field. Specifically, if $q = p^2$ for a prime power $p$, and if $d<q$, then the results in~\cite{guruswami2017efficiently} 
 provide an information super-set of size $p^t(p-1)$ for $S=p^t(p-1)-d(d+1)M-1$, where $M = m^d + O(m^{d-1})$ and $t=\left\lceil\log_p (dM)/2\right\rceil$. Note that the $q=2$ case is not handled in \cite{guruswami2017efficiently} due to the requirement $q=p^2$ for a prime power $p$.
%all subsets of $\cT$ of size $|\cT|-P$ each contain an information set for $\cC$. 

\section{Details of the greedy approach}\label{section:greedydetails}

In this section, we illustrate via an example a greedy method for constructing~$2$-information super-sets for binary code~$RM(1,m)$, for the case where~$m=8$. As described in Section~\ref{section:greedysetup}, we first pick the~$9$ index sets of size at most one, i.e., columns of~$G$ forming $I_9$. 
% We now show that~$u=5$ additional columns suffice to form a~$2$-information super-set.
% , i.e., the~$9$ columns of~$G$ of hamming weight at most one. 
To pick~$u$ additional columns, we set a weight limit~$\eta$ on these columns as a tunable parameter; that is, we require~$|S_i| \le \eta$ for all~$i \in [u]$, a practice that has proven to improve empirical performance. In this example we let~$\eta = 4$. 

\textit{Step 1}: We begin by setting~$S_1 \triangleq [\eta] = \{1,2,3,4\}$, which means $|T_1|=|T_2|=|T_3|=|T_4|=1$. This gives us an index set of size four; equivalently, we have picked the following additional column (ignoring the last row of~$G$):
\begin{align*}
    \begin{bmatrix}
    1 & 1 & 1 & 1 & 0 & 0 & 0 & 0
\end{bmatrix}^\intercal.
\end{align*}

\textit{Step 2}: From~$S_2$ and onward, at each step we pick the smallest~$i \in [m]$ such that~$|T_i| = 1$, and include~$i$ in the corresponding index set; in the example, at the start of the second step, $i=1$, so we include~$1$ in~$S_2$, which effectively means~$|T_1| = 2$. If~$|T_i| = 2$ for any~$i \in [m]$, we no longer include~$i$ in any future index set unless otherwise needed. If, after picking the smallest~$i$, the weight limit of the current index set is not yet reached, then we greedily include in the index set all~$j \in [m]$ such that~$|T_j| = 0$, until the limit is reached. In the example, $S_2 = \{1\}$, and $|T_5|=|T_6|=|T_7|=|T_8|=0$ at the moment, so we add~$5,6,7$ to~$S_2$, making~$S_2 = \{1,5,6,7\}$. Notice that we cannot include~$8$ because that will take us over the weight limit. We have thus far picked the following two columns:

\begin{align*}
    \begin{bmatrix}
    1 & 1 & 1 & 1 & 0 & 0 & 0 & 0 \\
    1 & 0 & 0 & 0 & 1 & 1 & 1 & 0
\end{bmatrix}^\intercal.
\end{align*}

\textit{Step 3}: Since the first row already has weight two, while~$|T_2|=|T_3|=|T_4|=1$, in the next step we begin by adding~$2$ to~$S_3$. Afterwards we include~$8$ because~$|T_8|=0$. Notice that we still have not reached the~$\eta=4$ limit even after including all~$j \in [m]$ such that~$|T_j| = 0$. When this is the case, we continue to greedily add the smallest~$i \in [m]$ such that~$|T_i| = 1$; in this case, $i=3$ is added, making~$S_3 = \{2,3,8\}$ at the moment:

\begin{align*}
    \begin{bmatrix}
    1 & 1 & 1 & 1 & 0 & 0 & 0 & 0\\
    1 & 0 & 0 & 0 & 1 & 1 & 1 & 0\\
    0 & 1 & 1 & 0 & 0 & 0 & 0 & 1
\end{bmatrix}^\intercal.
\end{align*}

Notice that~$T_2=T_3$ as of now, a conflict that can and needs to be addressed in future steps. Because~$S_3$ has one more weight to spare, one might be inclined to greedily include~$4$ and make~$S_3 = \{2,3,4,8\}$. However, this would mean~$T_2=T_3=T_4$, making the condition~$T_i \ne T_j$ for all distinct $i,j\in[0,m]$ difficult to satisfy. For this reason, we add~$i=5$, the next smallest~$i$ with~$|T_i|=1$, instead of~$4$. After three steps, we have picked the following three columns:

\begin{align*}
    \begin{bmatrix}
    1 & 1 & 1 & 1 & 0 & 0 & 0 & 0\\
    1 & 0 & 0 & 0 & 1 & 1 & 1 & 0\\
    0 & 1 & 1 & 0 & \textbf{1} & 0 & 0 & 1
\end{bmatrix}^\intercal.
\end{align*}

\textit{Step 4}: Because a conflict was introduced in Step three ($T_2=T_3$), we first address it by including~$2$, and not~$3$, in~$S_4$. Then, since there is no~$i \in [m]$ with~$|T_i|=0$, we simply add~$4,6,7$ to~$S_4$, making~$S_4=\{2,4,6,7\}$; again we cannot add~$8$ because we are out of weights to spare. We have picked the following four columns:

\begin{align*}
    \begin{bmatrix}
    1 & 1 & 1 & 1 & 0 & 0 & 0 & 0\\
    1 & 0 & 0 & 0 & 1 & 1 & 1 & 0\\
    0 & 1 & 1 & 0 & 1 & 0 & 0 & 1\\
    0 & 1 & 0 & 1 & 0 & 1 & 1 & 0
\end{bmatrix}^\intercal.
\end{align*}

Notice that the~$T_2=T_3$ conflict has been resolved; that being said, a new conflict~$T_6=T_7$ has emerged, and we need to address it in the next step. Also observe that~$|T_i|\ge 2$ holds for all but~$T_8$. 

\textit{Step 5}: We add~$6$ to~$S_5$ in order to resolve the new conflict, and we add~$8$ so that~$|T_8|=2$, all in one step, making~$S_5 = \{6,8\}$. We have picked the following five columns:

\begin{align*}
    \begin{bmatrix}
    1 & 1 & 1 & 1 & 0 & 0 & 0 & 0\\
    1 & 0 & 0 & 0 & 1 & 1 & 1 & 0\\
    0 & 1 & 1 & 0 & 1 & 0 & 0 & 1\\
    0 & 1 & 0 & 1 & 0 & 1 & 1 & 0\\
    0 & 0 & 0 & 0 & 0 & 1 & 0 & 1
\end{bmatrix}^\intercal.
\end{align*}

The weight limit is not exactly reached by~$S_5$, but it does not need to be. With the above choice of~$S_i$'s we have~$$T_0 = \{1,2,3,4,5\}, T_1 = \{1,2\}, T_2 = \{1,3,4\}, T_3 = \{1,3\}, T_4 = \{1,4\},$$ $$T_5 = \{2,3\}, T_6 = \{2,4,5\}, T_7 = \{2,4\}, T_8 = \{3,5\}.$$ Clearly, the~$T_i$'s satisfy the required conditions given in Section~\ref{section:greedysetup}; thus~$u=5$ index sets suffice. This is an improvement over Theorem~\ref{thm:twomplusone}, where~$u=m=8$.

Now we present Algorithm~\ref{algo:greedy}, which the above example illustrates. In what follows, we use the term~$2$-conflict to refer to the scenario where~$T_i = T_j$; such a conflict will be resolved in the next step (e.g., $T_2=T_3$ in the example). A~$3$-conflict is when three~$T_i$'s are equal; such a conflict is hard to resolve (e.g., $T_2=T_3=T_4$ in the example) and will be avoided.

\begin{algorithm}
  \caption{Greedy method for~$2$-information super-sets for~$RM(1,m)$}\label{algo:greedy}
  \Input{Positive integers~$m \ge 4$ and~$\eta \le m$.}
  \Output{Integer~$u \ge 2$, sets~$S_1,\ldots,S_u \subseteq [m]$.}
  $u \gets 1$\;
  $S_1 \gets [\eta]$\;
  $T_j \gets \{i \in [u]: j \in S_i\}$ for all~$j \in [m]$\;
  \While{$|T_i| < 2$ for some~$i \in [m]$ or~$T_i = T_j$ for some distinct~$i, j \in [m]$}{%
  $u \gets u + 1$\;
  $S_u \gets \emptyset$\;
  $C \gets$ the set of indices resolving any~$2$-conflicts\;
  $S_u \gets S_u \cup C$\;
  $T_j \gets \{i \in [u]: j \in S_i\}$ for all~$j \in [m]$\;
  $i \gets$ the smallest~$i' \in [m]$ such that~$|T_{i'}|=1$\;
  $S_u \gets S_u \cup \{i\}$\;
  $T_j \gets \{i \in [u]: j \in S_i\}$ for all~$j \in [m]$\;
  $Z_0 \gets \{i: |T_i| = 0\}$\;
  $D_0 \gets $ the set of~$\operatorname{min}(|Z_0|, \eta - |S_u|)$ smallest indices in~$Z_0$\;
  % Sort~$Z_0$ in ascending order\;
  % \While {$|S_u| < \eta$ and $Z_0$ is not empty}{%
  % $i \gets Z_0[0]$\;
  % $S_u \gets S_u \cup \{i\}$\;
  % $Z_0 \gets Z_0 \setminus \{i\}$\;
  % }
  $S_u \gets S_u \cup D_0$\;
  $T_j \gets \{i \in [u]: j \in S_i\}$ for all~$j \in [m]$\;
  $Z_1 \gets \{i: |T_i| = 1\}$\;
  $D_1 \gets $ the set of~$\operatorname{min}(|Z_1|, \eta - |S_u|)$ smallest indices in~$Z_1$ not creating a~$3$-conflict\;
  $S_u \gets S_u \cup D_1$\;
  $T_j \gets \{i \in [u]: j \in S_i\}$ for all~$j \in [m]$\;
  % Sort~$Z_1$ in ascending order\;
  % \While {$|S_u| < \eta$ and $Z_1$ is not empty}{%
  % $j \gets 0$\;
  % \While {$j < |Z_1|$ and $Z_1[j]$ causes a~$3$-conflict}{%
  % $j \gets j+1$\;
  % }
  % $i \gets Z_1[j]$\;
  % $S_u \gets S_u \cup \{i\}$\;
  % $Z_1 \gets Z_1 \setminus \{i\}$\;
  % }
  }
  $T_0 \gets \{j \in [u]: |S_j|$ is even\}\;
  \If {$|T_0| < 2$ or~$T_0 = T_i$ for some~$i \in [m]$}{%
  Abort\;
  }
\end{algorithm}

We note that Algorithm~\ref{algo:greedy} is not always guaranteed to give~$2$-information super-sets. For~$\eta = 2$, the algorithm will always succeed and give~$u=m$, matching the result of Theorem~\ref{thm:twomplusone}. Through a greedy computer search via Algorithm~\ref{algo:greedy}, we were able to obtain the results in Figure~\ref{fig:greedyresults} which outperform Theorem~\ref{thm:twomplusone}. 
% [Add pseudocode with an example.]

\section{Proof of Theorem~\ref{thm:elldms}}\label{section:thmsevenproof}

Suppose~$s\triangleq\log_2(S+1)$ is an integer. Let~$T(i,j)$ denote the number of times~$RM(i,j)$ appears in the tree. In what follows, we derive a general expression for~$T(i,j)$ which will be useful for obtaining our bound. We first make the following observations.

\begin{enumerate}
    \item From the recursive structure, we have~$T(i,j) = T(i,j+1)+T(i+1,j+1)$.
    \item For any leaf~$\ell$ resulting from this base case, the number of times~$\ell$ appears is exactly equal to the count of its parent in the tree; that is, $T(i,i+s) = T(i,i+s+1)$ for all~$2 \le i \le d$.
    \item Nodes on the rightmost path of the tree from the root to the leaf, i.e., nodes of the form~$RM(i,m-d+i)$, only appear once; that is, $T(i,m-d+i) = 1$ for all~$1\le i \le d$.
\end{enumerate}

We prove by induction that for all~$i \in[2,d]$, we have~$T(i,j) = \binom{m-j}{d-i}$ for all~$j \in [i+s+1,m-d+i]$. In detail, we first show~$T(d,j) = \binom{m-d-s-1}{d-d}$ as the initial case of the induction. Then, the induction assumption states that the expression holds for some~$i \in [2,d]$, and we then show that it holds for~$i-1 \in [2,d]$. In the initial case, we easily have~$T(d,j) = 1 = \binom{m-j}{d-d}$ for all~$j \in [d+s+1,m]$ from the third observation. Now, suppose~$T(i,j) = \binom{m-j}{d-i}$ for some~$i \in [2,d]$. Then, from the first and third observations,
\begin{align*}
    T(i-1,j) &= T(i-1,j+1)+T(i,j+1) \\ &= T(i-1,j+2)+T(i,j+2)+T(i,j+1) \\ &=\ldots = T(i-1,m-d+i-1)+\sum_{h=j+1}^{m-d+i-1} T(i,h) \\&= 1+\sum_{h=j+1}^{m-d+i-1} \binom{m-h}{d-i} \\&= 1+\sum_{h=d-i+1}^{m-j-1} \binom{h}{d-i} \\ &= \sum_{h=d-i}^{m-j-1} \binom{h}{d-i} \\&= \binom{m-j}{d-i+1},
\end{align*}
where the last step follows from the hockey-stick identity. This indeed proves that~$T(i,j) = \binom{m-j}{d-i}$.

In the first base case, we stop the breakdown when the condition~\eqref{equation:condition} is violated, i.e., when~$m-d=s$. This means the leaves from this base case are of the form~$RM(i,i+s)$ for~$ 2\le i \le d$. Recall that the minimum distance of~$RM(i,i+s)$ is~$d_{\textrm{min}}=2^s = S+1$, and therefore, the size of any~$S$-information super-set for~$RM(i,i+s)$ equals its block length~$2^{i+s}-(S+1)+S+1=2^{i+s}$. To count the number of times this base case is reached, i.e., the number of leaves from the first base case, notice that the expression for~$T(i,j)$ specifies to $T(i,i+s+1) = \binom{m-i-s-1}{d-i}$. This leads to~$T(i,i+s) = \binom{m-i-s-1}{d-i}$ from the second observation. Taking~$i$ from~$2$ to~$d$ results in the first part of the bound, i.e., $\sum_{i=2}^d \binom{m-i-s-1}{d-i} 2^{i+s}$.

To prove the second part, first notice that because~$RM(2,2+s)$ is a leaf node whose parent is~$RM(2,3+s)$, whose other child~$RM(1,2+s)$ is another leaf, the breakdown cannot go any further than this; in other words, $i \ge 2+s$. On the other hand, $i$ cannot be greater than~$m-d+1$, which corresponds to the last leaf node on the rightmost path of the tree. In addition, notice that leaf node~$RM(1,i)$ appears if and only if~$RM(2,i+1)$ appears, indicating that it suffices to count the number of nodes of the form~$RM(2,i+1)$ for all~$i \in [2+s,m-d+1]$. To this end, from previous analysis we have~$T(2,i+1) = \binom{m-i-1}{d-2}$, which means~$T(1,i) = \binom{m-i-1}{d-2}$. Putting all of this together yields the second part of the bound, i.e., $\sum_{j=2+s}^{m-d+1} \binom{m-j-1}{d-2} U(1,j,S)$.

In the case where~$s$ is not an integer, we stop the breakdown when~$m-d=\left\lceil s\right\rceil$. This means the leaves from the first base case are now of the form~$RM(i,i+\left\lceil s\right\rceil)$ for~$ 2\le i \le d$. The minimum distance of~$RM(i,i+\left\lceil s\right\rceil)$ is~$d_{\textrm{min}}=2^{\left\lceil s\right\rceil}$, and the size of any~$S$-information super-set for~$RM(i,i+\left\lceil s\right\rceil)$ equals~$2^{i+\left\lceil s\right\rceil}-2^{\left\lceil s\right\rceil}+S+1$. 
% This means the minimum distance of~$RM(d,m)$ now becomes~$d_{\textrm{min}}=2^{\left\lceil s\right\rceil}$, and, accordingly, the size of any~$S$-information super-set for~$RM(d,m)$ now equals~$2^{d+\left\lceil s\right\rceil}-2^{\left\lceil s\right\rceil}+S+1$. 
The counting argument follows the same logic.
\fi

\end{document}